\documentclass[a4paper,11pt]{article}
\usepackage[utf8]{inputenc}

\usepackage{amsthm,amsmath,amssymb,amsfonts,graphics}

\theoremstyle{plain}
\newtheorem{theorem}{Theorem}
\newtheorem{Pro}[theorem]{Proposition}
\newtheorem{Lem}[theorem]{Lemma}
\newtheorem{Cor}[theorem]{Corollary}
\newtheorem{Ex}{Example}

\theoremstyle{definition}

\newtheorem{definition}[theorem]{Definition}

\theoremstyle{remark}

\def\Z{\mathbb{Z}}	
\def\C{\mathbb{C}}	
\def\R{\mathbb{R}}	
\def\N{\mathbb{N}}	
\def\S{\mathcal{S}}

\def\cT{\mathcal{T}}
\def\cS{{\mathcal S}}
\def\bbbn{{\mathbb N}}
\def\bbbz{{\mathbb Z}}

\def\i{{\rm i}}
\def\xx{{\Sigma}}

\renewcommand{\leq}{\leqslant} 		
\renewcommand{\geq}{\geqslant}
\def\hp{{\hat{p}}}
\def\hq{{\hat{q}}}

\def\bbbc{{\mathbb C}}
\newcommand\cD{{\mathcal D}}

\newcommand\cH{{\mathcal H}}

\def\a{\alpha}
\def\b{\beta}
\def\om{\omega}

\def\d{\partial}
\def\fieldk{\bbbc}

\def\fA{\mathfrak{A}}

\def\ff{\mathfrak{f}}
\def\cA{\mathcal{A}}
\def\cT{\mathcal{T}}
\def\cS{{\mathcal S}}
\def\cZ{\mathcal{Z}}
\def\cH{\mathcal{H}}

\def\fI{\mathfrak{I}}
\def\fJ{\mathfrak{I}}

\def\pI{\pi_{\fI_a}}
\def\pJ{\pi_{\fI_b}}
\def\cI{\mathcal{I}}

\begin{document}

\bibliographystyle{unsrt}
\title{Quantisations of the Volterra hierarchy}
\author{Sylvain Carpentier $^\ddagger$, Alexander V. Mikhailov$^{\star}$ and 
Jing Ping 
Wang$ ^\dagger $
\\
$\ddagger$ QSMS, Seoul National University, South Korea, 
sylvain.carpentier23@gmail.com 
\\
$\star$ School of Mathematics, University of Leeds, UK, 
a.v.mikhailov@leeds.ac.uk\\
$\dagger$ School of Mathematics, Statistics \& Actuarial Science, University of 
Kent, UK, J.Wang@kent.ac.uk 
}
\date{}  

\maketitle
\begin{abstract}
In this paper we explore a recently emerged approach to
the problem of quantisation based on the notion of quantisation ideals. We 
explicitly prove that the nonabelian Volterra together with the whole hierarchy of its symmetries 
admit a deformation quantisation. We show that all odd-degree symmetries of the Volterra 
 hierarchy admit also a non-deformation quantisation. We  
discuss the quantisation problem for periodic Volterra hierarchy 
including their quantum Hamiltonians, central elements of the
  quantised algebras, and demonstrate  super-integrability of the quantum systems obtained. 
We show that the Volterra system with period $3$ admits a bi-quantum structure, 
which can be regarded as a quantum deformation of its classical bi-Hamiltonian structure.
\end{abstract}

\section{Introduction}

The problem of quantisation has a century long history. In 1925, inspired by  Heisenberg's commutation relations between coordinates and momenta \cite{Heisenberg25}, namely, 
\begin{equation}\label{qp}
 \hq_n\hp_m-\hp_m\hq_n=i\hbar \delta_{n,m},\quad \hq_n\hq_m-\hq_m\hq_n=0,\quad \hp_n\hp_m-\hp_m\hp_n=0, \qquad n,m=1,\ldots,N,
\end{equation}
Dirac proposed the concept of {\em quantum algebra} and noticed that in the 
limit $\hbar\to 0$ the commutators of observables are proportional to their  
Poisson brackets in classical mechanics $[\hq_n,\hp_m]\to i\hbar \{q_n,p_m\}$. 
He raised the issue of consistency of the commutation relations (\ref{qp}) with
each other and with the equations of motion for a finite Plank constant $\hbar\ne  0$
 \cite{Dirac25}. In fact, Dirac proposed the problem of non-commutative deformations of multiplication on Poisson 
manifolds that is presently an active research area. Important results in this direction have been obtained by Kontsevich \cite{Kontsevich2003}.
Witten, in his recent lectures \cite{Witten2021},  
pointed out that due to ``the operator ordering problem, there is no natural, general 
procedure to quantize a classical system'', and described some partial remedies 
to this problem. The general problem of quantisation is still open.

Recently, a fresh approach to the quantisation problem was proposed in 
\cite{AvM20}. It is proposed to start from a dynamical system defined on a 
free associative algebra  $\fA$ with a finite or infinite number of 
multiplicative generators. The dynamical system defines a derivation  
$\partial_t:\fA\mapsto\fA$. 
By quantisation it is understood a reduction of the dynamical
system on $\fA$ to the system defined on a quotient algebra 
$\fA_\fJ=\fA\diagup\fJ$ over a two-sided ideal 
$\fJ\subset\fA$ satisfying the following properties: 
\begin{enumerate}
 \item[({\rm i})] the ideal $\fJ$ is $\partial_t$--stable, that is, $\partial_t(\fJ)\subset\fJ$;
 \item[({\rm ii})] the quotient algebra $\fA_\fJ$ admits an additive basis of normally ordered monomials.
\end{enumerate}
In \cite{AvM20} an ideal satisfying the above two conditions is called  a
{\em 
quantisation ideal }, and $\fA_\fJ$ is called a {\em quantised algebra}.

The condition (i) is crucial. The reduction of a dynamical system 
corresponding to the derivation $\partial_t$ to the quotient algebra $\fA_\fJ$ 
is well defined if and only if the ideal is $\partial_t$--stable.

The second condition (ii)  enables one to define commutation relations 
between any two elements of the quotient algebra and uniquely represent  
elements of  $\fA_\fJ$ in the basis of normally ordered monomials 
(similar to a normal 
ordering in quantum physics). Finitely generated algebras, admitting a 
Poincar{\'e}--Birkhoff--Witt basis, and their quotients, satisfy the 
condition (ii). They have a wide range of applications, and share  some
properties with the commutative polynomial rings (see
\cite{berg, Lev} and references in).

Any finitely generated associative algebra can be presented as (is isomorphic to) a quotient of a 
free associative algebra over a suitable two-sided ideal. For example,  Dirac's 
quantum algebra is a quotient of the free algebra  $\bbbc\langle 
q_1,p_1,\ldots,q_N,p_N\rangle$ over the two-sided ideal generated by the 
commutation relations (\ref{qp}).

We emphasise that  quantisation proposed in \cite{AvM20} guarantee the consistency of
the ``commutation relations'' with each other and with the equations of motion 
(resolving the issue raised by Dirac) and the associativity of the 
non-commutative multiplication in the quantised algebra (which potentially could be an issue 
in the deformation quantisation). This new approach also results in examples of
non-deformation quantisations.

In order to apply this method of quantisation to a classical dynamical system
with commutative variables one needs to lift it to a system on a nonabelian
free associative algebra. Such   lifting is not unique (on the quantum level it has been
noted already by Dirac \cite{Dirac25}, and highlighted by Witten in his 
lectures \cite{Witten2021}). The guiding principle here is to preserve the most 
important properties of the classical system in the lifted one. For 
example, integrable systems admit  hierarchies of symmetries and we 
would like to have this property for the corresponding systems defined on a 
free associative algebras and for the quantised systems as well. Fortunately 
many integrable  systems admit such liftings \cite{EGR98,os98,miksok_CMP,ow2000, cw19-2},  and can be quantised by the
method 
proposed in \cite{AvM20}. Recently, the hierarchies of stationary Korteweg 
de--Vries equation and Novikov's equations have been quantised using the method 
of quantisation ideals \cite{BM2021}.

In this paper we study the  quantisation problem for the integrable 
nonabelian Volterra system 
\begin{equation}\label{vol}
\partial_{t_1 }(u_n)= \varrho K^{(1)}, \quad K^{(1)}=u_{n+1} u_n- u_n 
u_{n-1},\qquad n\in\bbbz
\end{equation}
and its hierarchy of symmetries. Here $\varrho\in\C$ is a constant which can be 
set to be equal to $1$ by the re-scaling  $u_n\to\varrho u_n$.  
In the classical (commutative) case system 
(\ref{vol}) was introduced by Zakharov,  Musher and  Rubenchik 
for 
the description of the fine structure of
the spectra of Langmuir oscillations in a plasma \cite{ZMR74}. Its 
integrability and Lax representation were discovered by Manakov 
\cite{Manakov74} 
and independently by Kac and van Moerbeke \cite{KacVanM75}.
The nonabelian version of the system (\ref{vol}), with variables 
$u_n(t_1 )$ taking values in a free associative algebra,  was studied by 
Bogoyavlensky \cite{Bog91}.

The Volterra system (\ref{vol}) is the first member of the infinite hierarchy 
of commuting symmetries
\[
\partial_{t_\ell}(u_n)=K^{(\ell)}(u_{n+\ell},\ldots ,u_{n-\ell}),\qquad 
\ell=1,2,\ldots,\ \ n\in\bbbz,
\]
where $K^{(\ell)}(u_{n+\ell},\ldots ,u_{n-\ell})$ are homogeneous polynomials of degree 
$\ell+1$ which can be found explicitly \cite{cw19-2}. The second member of 
the hierarchy 
\begin{equation}
\label{voltf2}
\partial_{t_2}(u_n)=K^{(2)}=u_{n+2}  u_{n+1}  u_n +u_{n+1}^2  u_n+  u_{n+1}  u_n^2 - 
u_n^2  u_{n-1}-u_n   u_{n-1}^2- u_n   u_{n-1}  u_{n-2}
\end{equation}
is given by the cubic polynomial.  It can be 
straightforwardly verified that 
$\partial_{t_2}(\partial_{t_1}(u_n))=\partial_{t_1}(\partial_{t_2}(u_n))$ and 
thus (\ref{voltf2}) is a cubic symmetry of (\ref{vol}).

In the new approach the quantisation problem for equation  (\ref{vol}) reduces 
to the problem of finding two-sided ideals in the free associative algebra  
$\fA=\fieldk\langle u_n\,;\, n\in\bbbz\rangle$ 
  generated by an infinite number of non-commuting variables such that the
above conditions (i) and (ii) are satisfied. It is obvious that the ideal $\fJ$ 
generated by the infinite set of polynomials 
\begin{equation}\label{ideal0}
\fJ= \langle u_nu_m -\omega_{n,m}u_mu_n\, ;\ n, m \in \mathbb{Z}, 
\omega_{n,m} \in \fieldk^*  
\rangle
\end{equation}
satisfies the condition (ii) for any choice of 
the parameters $\omega_{n,m}=\omega^{-1}_{m,n}$. In \cite{AvM20} it was stated 
that the ideal  $\fJ$ satisfies the condition (i) if and only if 
\[
 \omega_{n,n+1}=\omega_{n+1,n}^{-1}=\omega,\qquad  \omega_{n,m}=1 \ \ 
\mbox{if}\ \ |n-m|\geqslant 2. 
\]
Thus the quantisation ideal suitable for the Volterra system (\ref{vol}) is 
\begin{equation}\label{idi}
\fI_a= \langle \{ u_nu_{n+1}-\omega u_{n+1}u_n\,;\ n \in \mathbb{Z} \} \cup 
\{u_nu_m-u_mu_n\,;\ |n-m| >1 ,\ n,m \in\bbbz\ \} \rangle , 
 \end{equation}
leading to the commutation relations
\begin{equation}\label{comm1}
  u_nu_{n+1}=\omega u_{n+1}u_n,\qquad u_nu_m=u_mu_n\ \ 
\mbox{if}\ \ |n-m|\geqslant 2,\quad n,m \in\bbbz
\end{equation}
in the quotient algebra $\fA\diagup\fI_a$.
It was verified by direct computations that the ideal $\fI_a$ is invariant 
with respect to derivations defined by a few first symmetries of the Volterra 
hierarchy and conjectured that it is also true for the whole hierarchy. In this 
paper we give an explicit proof for the above conjecture (Theorem \ref{main1}). 
The ideal $\fJ_a$ corresponds to a 
deformation quantisation. In the limit $\omega\to 1$ it leads to the classical 
commutative case. 

 It was claimed in 
\cite{AvM20} that the cubic symmetry of the Volterra system, equation 
(\ref{voltf2}), admits two distinct quantisations 
ideals of the form (\ref{ideal0}). The first one coincides with $\fI_a$ defined 
by
(\ref{idi}), while the second one is 
\begin{equation}\label{idj}
\fI_b= \langle \{ u_nu_{n+1}-(-1)^n \omega u_{n+1}u_n\,;\,n \in \mathbb{Z} \} 
\cup 
\{u_nu_m+u_mu_n\,;\,|n-m| >1,\  n,m \in \mathbb{Z}\} \rangle\,.
\end{equation}
Note that the quantisation corresponding to the ideal $\fI_b$ is not a 
deformation of a commutative or Grassmann algebra. It is a new and {\em 
non-deformation} quantisation of equation (\ref{voltf2}) with the commutation 
relations
\begin{equation}\label{comm2}
 u_nu_{n+1}=(-1)^n \omega u_{n+1}u_n,\qquad  
u_nu_m+u_mu_n=0\ \ 
\mbox{if}\ \ |n-m|\geqslant 2,\quad n,m \in\bbbz
\end{equation}
in the quotient algebra $\fA\diagup\fI_b$. 
The ideal $\fI_b$ given by (\ref{idj}) is not invariant with respect to the 
Volterra 
system (\ref{vol}) and thus it is not suitable for its quantisation. In 
\cite{AvM20} it was claimed that the ideal $\fI_b$ is 
invariant with respect to a first few odd degree symmetries of the Volterra 
equation. In this paper we  prove  that the ideal $\fI_b$ 
(\ref{idj}) is a quantisation ideal for all odd degree members of the Volterra 
hierarchy (Theorem \ref{main2}).

In the quantum theory we replace real valued commutative variables $u_n$ by 
Hermitian elements. Their commutation relations are defined by the quantisation
ideal, which should be stable with respect to the Hermitian conjugation
(Definition \ref{herm}). In the case of the ideals $\fI_a$ and $\fI_b$, it implies
that $\omega=e^{2\i\hbar}$, where $\hbar$ is an arbitrary real parameter, an analogue of the Plank constant, and $\i^2=-1$.
Moreover, in the quantised equations of the Volterra hierarchy, we should introduce the factors $e^{\i\ell\hbar}$ which make
 the right-hand side of the equations self-adjoint, that is,
\begin{equation}\label{qVh}
\d_{t_\ell}(u_n)=e^{\i\ell\hbar}K^{(\ell)}(u_{n+\ell},\ldots ,u_{n-\ell}),\qquad
\ell=1,2,\ldots,\ \ n\in\bbbz .
\end{equation}

 In the algebra 
$\fA_{\fI_a}$ with commutation relations (\ref{comm1}) the  quantised Volterra 
equation   and its symmetry 
 can be represented in the Heisenberg form
\begin{eqnarray}
 &&   
\d_{t_1}(u_n)=e^{\i\hbar}K^{(1)}=\dfrac{\i}{2\sin(\hbar)}[H_1,u_n]
 \label{hei2},\\
&&    
\d_{t_2}(u_n)=e^{2\i\hbar}K^{(2)}=\dfrac{\i}{2\sin(2\hbar)}[H_2,u_n],\label{hei3}
\end{eqnarray}
where
\[
H_1=\sum\limits_{k\in\Z}u_k
\qquad H_2=
    \sum\limits_{k\in\Z}(u_k^2+u_{k+1}u_k+u_ku_{k+1}).
\]

In the algebra $\fA_{\fI_b}$ with commutation relations (\ref{comm2}), the 
first member of the quantised Volterra sub-hierarchy of odd degree symmetries
has the same Heisenberg form (\ref{hei3}). Moreover, in the case of 
the algebra  $\fA_{\fI_b}$ we have $H_2=H_1^2$, which is 
not true for the algebra $\fA_{\fI_a}$. 

The quantisation of the Volterra system  was studied by 
Volkov and Babelon in the frame of the quantum inverse scattering method \cite{volkov,babel}.  In the paper by Inoue and Hikami \cite{InKa},  the commutation relations (\ref{comm1}), as well as a first few Hamiltonians of the classical and quantum Volterra hierarchy were found using ultra-local Lax representation and  $R$--matrix technique. Our alternative approach does not rely on the existence of a Lax or Hamiltonian structures, and it enables us to reproduce the results presented in \cite{InKa} and to find a non-deformation quantisation (\ref{comm2}) for odd degree members of the Volterra hierarchy which is new and rather surprising.

The Volterra equation and its hierarchy admit periodic reductions with 
arbitrary positive integer period $M\in\mathbb{N}$. The periodic reduction is 
the identification $u_{n+M}=u_n$ for all $n\in\Z$. 
It reduces the infinite system of equations (\ref{vol}) to a system of $M$ 
equations on a finitely generated free algebra $\fA_M=\C\langle 
u_1,\ldots,u_M\rangle$. 
The problem of quantisation of the periodic Volterra hierarchies is
discussed in Section \ref{periodic}. In particular,  we show that the
Volterra system with period $3$ admits bi-quantum structure, which is a quantum 
analogue of its bi-Hamiltonian structure in the classical case.  In the
case $M=4$ we obtain three possible quantisations, and show that the obtained quantised
systems are super-integrable, whose first integrals and
central elements are explicitly presented.

\section{Integrable nonabelian Volterra hierarchy}\label{sec2}

In this section we introduce some basic notations required for this paper, and present
the Volterra hierarchy on a free associative algebra in an explicit form.

Let $\fA=\fieldk\langle u_n\,;\, n\in\bbbz\rangle$ be a free associative 
algebra generated by an infinite number of non-commuting variables.   
There is a natural automorphism $\cS\,:\,\fA\mapsto\fA$, 
which we call the {\em shift operator}, defined as 
\[
 \cS: a( u _k,\ldots , u _r)\mapsto a( u _{k+1},\ldots , u _{r+1}),\quad 
\cS:\alpha\mapsto\alpha, \qquad a( u _k,\ldots , u _r)\in \fA,\ \ 
\alpha\in\fieldk.
\]
Thus $\fA$ is a difference algebra. Let $\cT$ denote the antiautomorphism of 
$\fA$ defined by
\[
 \cT(u_k)=u_{-k},\quad \cT(a\cdot b)=\cT(b)\cdot\cT(a),\quad 
\cT(\alpha)=\alpha, \ \ \ a,b\in\fA,\ \ \ \alpha\in\fieldk.
\]
The involution $\cT$ is a composition of the reflection in the alphabet index 
$u_k\mapsto u_{-k}$ and the transposition of the monomials. For example:
\[
 \cT(u u_1+u_4u_1u_{-3}u_{-2})=u_{-1}u+u_{2}u_{3}u_{-1}u_{-4}.
\]

A derivation  $\cD  $ of the algebra  $\fA$ is a $\C$--linear map
satisfying   Leibniz's rule  
\[\cD  (\alpha a+\beta  b)=\alpha\cD  (a)+\beta\cD  (b),\qquad  \cD  (a\cdot 
b)=\cD (a)\cdot b+a\cdot\cD (b),\qquad  a,b\in\fA,\ \ \alpha,\beta\in\C.\] 
Thus a derivation $\cD $ can be uniquely defined by its action on the 
generators 
and $\cD  (\alpha)=0,\ \alpha\in\C$.

A derivation $\cD $ is called evolutionary if it commutes with the automorphism 
$\cS$. An evolutionary derivation is completely characterised by its action on 
the generator $u$ (we often write $u$ instead of $u_0$), that is,
\[
 \cD  (u)=a\quad \mbox{and} \quad \cD (u_k)=\cS^k (a),\qquad a\in \fA.
\]
Thus it is natural to adopt the notation $\cD_a$, such that $\cD_a(u)=a$, for 
an evolutionary derivation with the characteristic $a$.
A commutator of evolutionary derivations $\cD_a,\cD_b$ is also the evolutionary 
derivation   $[\cD_a,\cD_b]=\cD_c$ with the characteristic 
$c=\cD_a(b)-\cD_b(a)$, 
which is  called the Lie bracket  of the elements $a$ and $b$.  Evolutionary 
derivations form a Lie subalgebra of the Lie algebra  of derivations of $\fA$.

Assuming that the generators $u_k$ depend on $t\in\C$ we can identify an 
evolutionary $\cD_a$ with an infinite system of differential-difference 
equations
\[
 \d_t(u_n)=\cD_a(u_n)=\cS^n(a),\qquad n\in\Z.
\]
Therefore we can say that $\d_t(u)=a$ defines a derivation of $\fA$.

The Volterra 
system (\ref{vol}) defines the derivation $\d_{t_1}\,:\,\fA\mapsto\fA,$ which
commutes with the automorphism and anti-commute with the involution $\cT$, i.e.,
\[ 
 \cS\cdot \d_{t_1}=\d_{t_1}\cdot \cS,\qquad \cT\cdot \d_{t_1}=-\d_{t_1}\cdot \cT\, .
\]
The differential-difference system (\ref{voltf2}) defines another evolutionary 
derivation $\d_{t_2}$ commuting with $\cS$ and anti-commuting with $\cT$.
Evolutionary derivations commuting with $\d_{t_1}$ are symmetries of the Volterra
system. It can be straightforwardly verified that $[\d_{t_1},\d_{t_2}]=0$ and
thus equation (\ref{voltf2}) is a symmetry of the Volterra system.

It is well known that the Volterra system has an infinite hierarchy of 
commuting symmetries.  They can be found using Lax representations both in 
commutative \cite{Manakov74} and non-commutative \cite{Bog91} cases, or the 
recursion operators \cite{wang12, cw19-2}. Remarkably, the explicit expressions 
for generalised symmetries of the Volterra system (\ref{vol}) can be presented 
in terms of a family of nonabelian homogeneous difference polynomials  
\cite{cw19-2}, which is inspired by 
the polynomials in the commutative case discovered in \cite{svin09,svin11}.

Let us assume that the generators  $u_k$ of the free associative algebra $\fA$ 
depend on an infinite set of ``times'' $t_1,t_2, \ldots$ . It follows from  
\cite{cw19-2} that the hierarchy of commuting symmetries of the Volterra system 
(\ref{vol}) can be written in the following explicit form
\begin{equation}\label{volh}
 \d_{t_\ell}(u)=\cS(X^{(\ell)})u-u\cS^{-1}(X^{(\ell)}),\qquad \ell\in\N\, ,
\end{equation}
where the (noncommutative) polynomials $X^{(\ell)}$ are given by explicit 
formulae
\begin{equation}\label{xl}
 X^{(\ell)}=
 \sum_{0\leq \lambda_{1}\leq \cdots \leq \lambda_{\ell}\leq \ell-1}
	\left(\prod_{j=1}^{\rightarrow{\ell}} u_{\lambda_{j}+1-j} \right).
\end{equation}
Here $\prod_{j=1}^{\rightarrow{\ell}}$   denotes the order of the values $j$, 
from $1$ to $\ell$ in the product of the noncommutative generators 
$u_{\lambda_{j}+1-j}$.
For example, we have $X^{(1)}=u$ and
\begin{eqnarray}
 &&\hspace{-1cm}X^{(2)}=u_1 u +u^2+u u_{-1};\label{x1}\\
 &&\hspace{-1cm} X^{(3)}=u_2 u_1 u+u_1^2 u + u u_1 u +u_1 u^2+u^3+ u u_{-1} 
u+u_1 u u_{-1}+u^2 u_{-1}+uu_{-1}^2+u u_{-1} u_{-2}.\label{x2}
\end{eqnarray}
Note that $\cT(X^{(\ell)})=X^{(\ell)}$, and thus we have
$\cT\cdot \d_{t_\ell}=-\d_{t_\ell}\cdot \cT$ for all $\ell$.
Clearly, we get the Volterra equation (\ref{vol}) when $\ell=1$ and the system 
(\ref{voltf2}) when $\ell=2$.

\section{Quantisation ideals of the Volterra equation and its 
symmetry}\label{sec3}
In this section, we prove the statements on quantisation ideals for the 
Volterra equation (\ref{vol})
itself and its symmetry (\ref{voltf2}) stated in \cite{AvM20}.

Let  $\fJ\subset\fA$ be a two-sided ideal generated by the infinite set of 
polynomials $\ff_{i,j}$:
\begin{equation}\label{ideal00}
 \fJ=\langle \ff_{i,j}\,;\, i<j,\ i,j\in\Z\rangle,\qquad 
 \, \ff_{i,j}=u_i u_j-\omega_{i,j}u_j u_i,
\end{equation}
where $\omega_{i,j}\in\C^*$ are arbitrary non-zero complex parameters.
Given an ideal $\fI$, we denote the projection on the quotient algebra by
by $\pi_{\fJ}: \fA \rightarrow \fA/\fJ $.
The quotient algebra $\fA\diagup\fJ$ has an additive basis of  {\em standard 
normally ordered monomials}
\[
 u_{i_1}u_{i_2}\cdots u_{i_n}\, ;\qquad i_1\geqslant 
i_2\geqslant\cdots\geqslant i_n,\ i_k\in\Z,\ n\in\N.
\]
Indeed, in $\fA\diagup\fJ$ any polynomial can be represented in this basis by 
recursive replacements 
$u_n u_m\to \omega_{n,m}u_m u_n$ if $m>n$ in the monomials. Thus the condition 
(ii) for the ideal $\fJ$ is satisfied. The condition (i) imposes constraints on 
the structure constants $\omega_{n,m}$ of the ideal. 

\begin{Pro}\label{proV}
 The ideal $\fJ$ (\ref{ideal00}) is invariant with respect to the Volterra 
dynamics (\ref{vol}) if and only if
\[
 \omega_{n,n+1}= \omega_{0,1},\qquad  \omega_{n,m}=1 \ \ 
\mbox{if}\ \ m-n\geqslant 2,\qquad n,m\in\Z. 
\]
\end{Pro}
Denoting $\omega_{0,1}=\omega$, we arrive to the commutation relations 
(\ref{comm1}) and the ideal $\fI_a$ given by (\ref{idi}).
\begin{proof}
 Let us differentiate $\ff_{i,j}$ ($i<j$) by the derivation $\partial_{t_1}$ 
associated to the Volterra equation (\ref{vol}). We have
 \begin{eqnarray*}
 &&\partial_{t_1}\left(\ff_{i,j}\right) = u_{i+1}u_{i}u_j-u_i u _{i-1}u_j + 
u_iu_{j+1}u_j-u_iu_ju_{j-1} \\
&&\qquad -\omega_{i,j}(u_{j+1}u_{j}u_i-u_j u _{j-1}u_i + 
u_ju_{i+1}u_i-u_ju_iu_{i-1} ).
 \end{eqnarray*}
We project this equation on the quotient algebra and require
\begin{eqnarray}
&&0=\pi_{\fI}\left(\partial_{t_1}(\ff_{i,j})\right) = \omega_{i,j} 
(\omega_{i+1,j}-1)u_ju_{i+1}u_{i}+ \omega_{i,j} (1-\omega_{i-1,j}) u_ju_i u 
_{i-1}\nonumber\\
&&\qquad +  \omega_{i,j} (\omega_{i,j+1}-1)u_{j+1}u_ju_i+ \omega_{i,j} 
(1-\omega_{i,j-1})u_ju_{j-1}u_i, \label{pjv}
\end{eqnarray}
where we use the convention $\omega_{i,i}=1$.
When $j > i+2$, the four monomials $u_{j+1}u_ju_i$, $u_ju_iu_{i-1}$, 
$u_ju_{i+1}u_i$ and $u_ju_{j-1}u_i$ are linearly independent. Thus 
$\pi_{\fI}\left(\partial_{t_1}(\ff_{i,j})\right) =0$ if and only if all their 
coefficients vanish since $\omega_{i,j}\neq 0$. 
This leads to 
$$\omega_{i+1,j}=\omega_{i-1,j}=\omega_{i,j+1}=\omega_{i,j-1}=1.$$
Hence we must have $\omega_{i,j} = 1$ whenever $i+1 < j$. Using this result, it 
follows from (\ref{pjv}) that
\begin{eqnarray*}
0=\pi_{\fI}\left(\partial_{t_1}(\ff_{i,i+2})\right) = \omega_{i,i+2} 
(\omega_{i+1,i+2}-\omega_{i,i+1})u_{i+2}u_{i+1}u_{i}.
\end{eqnarray*}
This implies that all the $\omega_{i,i+1}$ are equal to each other. Let 
$\omega=\omega_{i,i+1}$. It remains to check that (\ref{pjv}) is valid for 
$j=i+1$. Indeed, 
\begin{eqnarray*}
&&\pi_{\fI}\left(\partial_{t_1}(\ff_{i,i+1})\right) =  
\omega(1-\omega_{i-1,i+1}) u_{i+1}u_i u _{i-1} +  \omega 
(\omega_{i,i+2}-1)u_{i+2}u_{i+1} u_i=0,
\end{eqnarray*}
and we proved the statement.
\end{proof}
\begin{Pro}\label{proV2}
 The ideal $\fJ$ (\ref{ideal00}) is invariant with respect to the  dynamical 
system (\ref{voltf2}), i.e.,
 $\d_{t_2}(u)=\cS(X^{(2)})u-u\cS^{-1}(X^{(2)})$ only in two cases:
 \begin{enumerate}
  \item[{\rm (a)}.]
$\qquad
 \omega_{n,n+1}= \omega,\qquad  \omega_{n,m}=1 \ \ 
\mbox{if}\ \ m-n\geqslant 2,\qquad n,m\in\Z;
$
  \item[{\rm (b)}.]
$\qquad
 \omega_{n,n+1}= (-1)^{n}\omega,\qquad  \omega_{n,m}=-1 \ \ 
\mbox{if}\ \ m-n\geqslant 2,\qquad n,m\in\Z, 
$
 \end{enumerate}
where $\omega\in\C^*$ is an arbitrary non-zero complex parameter.
\end{Pro}
Thus, equation  (\ref{voltf2}) admits the same quantisation $\fA\diagup\fI_a$ 
(\ref{idi}) as the Volterra system. Additionally, it admits the quantisation 
with the ideal $\fI_b$ (\ref{idj}), which is not invariant with respect to the 
Volterra system (\ref{vol}). The latter quantisation is not a deformation of a 
commutative system.

\begin{proof} We  differentiate $\ff_{i,j}$ ($i<j$) by the derivation 
$\partial_{t_2}$ defined by equation (\ref{voltf2}) and project on the quotient 
algebra. When $i+2\leq j$ we have 
\begin{eqnarray}
&&\omega_{i,j}^{-1}\pi_{\fI}\left(\partial_{t_2}(\ff_{i,j})\right)=(\omega_{i+1,
j}  \omega_{i+2,j}-1) u_ju_{i+2}u_{i+1}u_i
+  ( \omega_{i+1,j}^2-1) u_ju_{i+1}^2 u_i
\nonumber\\&&\quad
+   (\omega_{i,j}  \omega_{i+1,j}-1) u_ju_{i+1}u_i^2
-  (\omega_{i,j}  \omega_{i-1,j}-1) u_ju_i^2 u_{i-1}
-   (\omega_{i-1,j}^2  -1) u_ju_iu_{i-1}^2
\nonumber\\&&\quad
-   (\omega_{i-1,j}  \omega_{i-2,j}-1) u_ju_iu_{i-1}u_{i-2}
+(\omega_{i,j+1}\omega_{i,j+2}-1) u_{j+2}u_{j+1}u_ju_i
\nonumber\\&&\quad
 + (\omega_{i,j+1}^2-1) u_{j+1}^2 u_ju_i
 +  (\omega_{i,j}\omega_{i,j+1}-1) u_{j+1}u_j^2 u_i
 -  (\omega_{i,j}\omega_{i,j-1}-1) u_j^2 u_{j-1}u_i
\nonumber\\&&\quad 
 - (\omega_{i,j-1}^2-1) u_ju_{j-1}^2 u_i
 - ( \omega_{i,j-1}\omega_{i,j-2}-1) u_ju_{j-1}u_{j-2} u_i, \label{pjv2}
\end{eqnarray}
where we use the convention $\omega_{i,i}=1$.
If $i+3 < j$ all monomials in (\ref{pjv2}) are distinct and one deduces from 
$\pi_{\fI}\left(\partial_{t_2}(\ff_{i,j})\right)=0$ that 
\begin{equation*}
\begin{split} 
\omega_{i+1,j}  \omega_{i+2,j}  &= \omega_{i+1,j}^2 =   \omega_{i,j}  
\omega_{i+1,j} =  \omega_{i,j}  \omega_{i-1,j} =   \omega_{i-1,j}^2 =  
\omega_{i-1,j}  \omega_{i-2,j} \\
=\omega_{i,j+1}\omega_{i,j+2}  &=  \omega_{i,j+1}^2 = \omega_{i,j}\omega_{i,j+1}
 =  \omega_{i,j}\omega_{i,j-1} =  \omega_{i,j-1}^2=  
\omega_{i,j-1}\omega_{i,j-2} = 1
 \end{split}
 \end{equation*}
 It follows that $\omega_{i,j} = \epsilon$ for all $i+1 < j$ where $\epsilon = 
\pm 1$. Next let us look at $\partial_{t_2}( \ff_{i,i+3})$. When $j=i+3$, 
(\ref{pjv2}) becomes
 \begin{eqnarray*}
&&\epsilon \pi_{\fI}\left(\partial_{t_2}(\ff_{i,i+3})\right)=\epsilon 
(\omega_{i+2,i+3}-\omega_{i,i+1}) u_{i+3}u_{i+2}u_{i+1}u_i, 
\end{eqnarray*}
which leads to $\omega_{i,i+1} = \omega_{i+2,i+3}$ for all $i \in \Z$. So the 
ideal is invariant under the automorphism $\cS^2$.
We now look at $\partial_{t_2}( \ff_{i,i+2})$. Substituting  $j=i+2$ into 
(\ref{pjv2}), we get
\begin{eqnarray*}
&&\epsilon \pi_{\fI}\left(\partial_{t_2}(\ff_{i,i+2})\right)=(\omega_{i+1,i+2}  
-\epsilon\omega_{i,i+1}) u_{i+2}^2 u_{i+1}u_i
\nonumber\\&&\qquad
+  ( \omega_{i+1,i+2}^2-\omega_{i,i+1}^2) u_{i+2} u_{i+1}^2u_i
+   (\epsilon  \omega_{i+1,i+2}-\omega_{i,i+1}) u_{i+2}u_{i+1}u_i^2 ,
\end{eqnarray*}
which vanishes if and only if $\omega_{i,i+1}=\epsilon  \omega_{i+1,i+2}$. 
Combining all the constraints obtained on $\omega_{i,j}$, we obtain the two 
cases listed in the statement.
Finally, we check
\begin{eqnarray*}
&&\omega_{i,i+1}^{-1}\pi_{\fI}\left(\partial_{t_2}(\ff_{i,i+1})\right)=(\omega_{
i,i+1} \epsilon- \omega_{i+1,i+2}) u_{i+2} u_{i+1}^2 u_i
-  (\omega_{i,i+1}  \epsilon-\omega_{i-1,i}) u_{i+1} u_i^2 u_{i-1} =0.
\end{eqnarray*}
Thus we complete the proof.
\end{proof}
In section \ref{secproof} we will show that
 every member of the Volterra hierarchy 
(\ref{volh}) admits the quantisation $\fA\diagup\fI_a$ (Theorem \ref{main1}) 
and that
 every even member of the Volterra hierarchy 
 $$\d_{t_{2\ell}}(u)=\cS(X^{(2\ell)})u-u\cS^{-1}(X^{(2\ell)}),\qquad \ell\in\N$$
 also admits the quantisation $\fA\diagup\fI_b$ (Theorem \ref{main2}).

In the classical commutative case the variables $u_n$ are usually assumed to be 
real valued. Thus, in the quantum case they should be presented by self adjoint 
operators with respect to the Hermitian conjugation $\dagger$. 
\begin{definition}\label{herm}
 The Hermitian conjugation $\dagger$ in algebra  $\fA$ is defined by the following rules
\[
 u_n^\dagger=u_n,\quad \alpha^\dagger=\bar{\alpha},\quad (a+b)^\dagger=a^\dagger+ b^\dagger,\quad (ab)^\dagger=b^\dagger a^\dagger,\qquad u_n,a,b\in\fA,\ \ \alpha\in\bbbc,
\]
where $\bar{\alpha}$ is the complex conjugate of $\alpha\in \bbbc$.
\end{definition}

The  algebra $\fA$ is $\bbbz_2$-graded as a linear space. It can be represented as a direct sum of self-adjoint and anti-self-adjoint subspaces
\[
 \fA=\fA^{+}\bigoplus\fA^{-},\qquad \fA^+=\{a\in\fA\,;\, a^\dagger=a\},\quad 
 \fA^-=\{a\in\fA\,;\, a^\dagger=-a\}\, .
\]
The Hermitian conjugation $\dagger$ can be extended to the quantised algebra $\fA \diagup\fJ$ if the ideal $\fJ$ is $\dagger$-stable: $\fJ^\dagger=\fJ$.  
\begin{Pro}\label{omegadag}
 The quantisation ideals $\fJ_a$ (\ref{idi}) and $\fJ_b$ (\ref{idj}) 
 are $\dagger$--stable if and only if $\omega^\dagger=\omega^{-1}$.
\end{Pro}
\begin{proof}
 Indeed, in the case of the ideal $\fJ_a$ we have
 \[
  (u_n u_{n+1}-\omega u_{n+1}u_n)^\dagger=u_{n+1}u_n-\omega^\dagger u_n u_{n+1}=-\omega^\dagger (u_n u_{n+1}-(\omega^\dagger)^{-1} u_{n+1}u_n)\in\fJ_a\Leftrightarrow \omega^\dagger=\omega^{-1}.
 \]
In the case for $\fJ_b$, the proof is similar.
\end{proof}

It suggests to represent $\omega=q^2,\ q=e^{\i \hbar}$, where $\hbar\in\R$ is a real constant (an analog of the Plank constant). Thus $(u_{n+1} u_n)^\dagger=u_n u_{n+1}=q^2 u_{n+1} u_n$. The quantum Volterra hierarchy, which is consistent with the condition $u_n^\dagger=u_n$, can be presented in the form
\begin{equation}\label{qvolh}
  u_{t_1}=q(u_1u-uu_{-1}),\qquad u_{t_\ell}= q^\ell\left(\cS(X^{(2\ell)})u-u\cS^{-1}(X^{(2\ell)})\right), \quad \ell\in\N.
\end{equation}

Finally, we present the Volterra system and its first symmetry in the 
Heisenberg form in the quotient algebras.
 In the algebra $\fA\diagup\fI_a$ with commutation relations (\ref{comm1}) the 
Volterra
equation (\ref{vol}) and its symmetry (\ref{voltf2}) can be represented in the 
Heisenberg form
\begin{equation}\label{qvola}
\begin{array}{ll}
    \d_{t_1}(u_n)=\dfrac{1}{q^{-1}-q}[H_1,u_n],\qquad 
&H_1=\sum\limits_{k\in\Z}u_k;\\
    \d_{t_2}(u_n)=\dfrac{1}{q^{-2}-q^2}[H_2,u_n],&H_2=
    \sum\limits_{k\in\Z}(u_k^2+u_{k+1}u_k+u_ku_{k+1}),
  \end{array}
\end{equation}
where $H_1$ and $H_2$ are self-adjoint algebraically independent and commuting Hamiltonians
$[H_1,H_2]=0$ in $\fA\diagup\fI_a$.

The quantisation $\fA\diagup\fI_b$ with commutation relations (\ref{comm2}) 
also enables us to present equation (\ref{voltf2}) in the Heisenberg form
\begin{equation}\label{qvolb}
  \d_{t_2}(u_n)=\dfrac{1}{q^{-2}-q^2}[H_2,u_n].
\end{equation}
Note that in the quantised algebra $\fA\diagup\fI_b$ we have $H_2=H_1^2$ and $H_2^\dagger=H_2$.

\section{Periodic Volterra hierarchy}\label{periodic}

In the Volterra system (\ref{vol}) we can assume that the function $u_n(t_1)$ 
is periodical in $n$ with an integer period $M\in\bbbn$, that is, $u_n=u_{n+M},\
n\in\Z$. 
In this case the infinite dimensional system (\ref{vol}) reduces to the 
$M$-dimensional dynamical system on  $\fA_M=\C\langle u_1,\ldots 
u_M\rangle=\fA/\cI_M$, 
where the ideal $\cI_M=\langle u_n-u_{n+M}\, ;\, n\in\Z\rangle$. The ideal 
$\cI_M$ is obviously stable with respect to evolutionary derivations.
We can take $u_n,\ n=1,\ldots M$ as canonical representatives of the cosets 
$u_k+\cI_M,\ k\in\Z$.  
The algebra $\fA_M$ is a difference algebra with the induced automorphism  $\cS 
(u_k)=u_{(k+1)\,{\rm mod}\, M}$ of order $M$.

The hierarchy of symmetries (\ref{volh}) of the Volterra system (\ref{vol}) 
reduces to the hierarchy of symmetries of the $M$-periodic system provided we 
count the subscript $k$ in $u_{k}$ modulo $M$. The cases $M=1,2$ lead to
trivial equations.

In the case $M=3$ the periodic Volterra system takes the form
\begin{equation}\label{vol3}
 \begin{array}{l}
 \d_{t_1}(u_{1})=u_2u_1-u_1u_3,\  \d_{t_1}(u_{2})=u_3u_2-u_2u_1,\  
\d_{t_1}(u_{3})=u_1u_3-u_3u_2\, .
\end{array}
\end{equation}
It has an infinitely hierarchy of commuting symmetries:
\[
\begin{array}{ll}
\d_{t_2}(u_{1})&=u_1^2 u_3+u_1  u_3  u_2+u_1  u_3^2-u_2  u_1^2-u_2^2  u_1-u_3  
u_2  u_1,\\&\\
\d_{t_3}(u_{1})&=u_1^3  u_3+u_1 ^2  u_3  u_2+u_1 ^2  u_3 ^2+u_1  u_2  u_1  
u_3+u_1  u_3  u_1  u_3
+u_1  u_3  u_2 ^2\\&+u_1  u_3  u_2  u_3+u_1  u_3^2  u_2+u_1  u_3 ^3-u_2  
u_1^3-u_2  u_1  u_2  u_1-u_2  u_1  u_3  u_1\\&-u_2^2  u_1^2-u_2 ^3  u_1-u_2  
u_3  u_2  u_1-u_3  u_2  u_1 ^2-u_3  u_2 ^2  u_1-u_3  ^2  u_2  u_1\, ,
\\ \cdots &
\end{array}
\]
For any $M$ the nonabelian Volterra hierarchy has a common  first integral 
$H=\sum\limits_{k=1}^M u_k$.

In the case of the finitely generated free algebra $\fA_M$ we consider more
general inhomogeneous ideals $\fJ_M\subset\fA_M$ (than  (\ref{ideal0})) 
generated by the polynomials $\ff_{i,j}$:
\begin{eqnarray}
&&\fJ_M=\langle  \ff_{i,j}, 1\le i<j\le M, i,j \in \mathbb{N} \rangle, \quad  
\ff_{i,j}=u_i u_j-\omega_{i,j}u_j u_i-   \sigma^r_{i,j}u_r -\eta_{i,j}, 
\label{fijm} 
\end{eqnarray}
where $\omega_{i,j}\ne 0,\ \omega_{i,j},\sigma^r_{i,j},\eta_{i,j}\in\C$ and we 
use Einstein summation convention, namely $\sigma^r_{i,j}u_r$ 
denotes $\sum\limits_{r=1}^{M}\sigma^r_{i,j}u_r$.
In this section, we explore the quantisation problem for periodic reductions of 
the Volterra system and its cubic symmetry.

\subsection{Quantisation of the periodic Volterra system}\label{finitev}
Similarly to what we did in Section \ref{sec3}, we are able to prove the 
following statement for the periodic Volterra equation:
\begin{theorem}\label{propM}
A nonabelian periodical Volterra chain with  period  $M$ admits a 
$\fJ_M$--quantisation if and only if the following commutation relations hold:
  \begin{eqnarray}
   M=3:&& u_nu_{n+1}=\alpha u_{n+1} u_n+\beta (u_1+u_2+u_3)+\eta,\ \ n\in
\Z_3;\label{v3}\\
  M=4:&&  u_1u_{2}=\alpha u_{2} u_1+\beta u_{2}+\gamma u_{1}-\beta 
\gamma,\label{v4}\\\nonumber && u_1 u_3=u_3 u_1-\beta u_2+\beta u_4,\\\nonumber 
&
 &u_4u_{1}=\alpha u_{1} u_4+\beta u_{4}+\gamma u_{1}-\beta \gamma,\\\nonumber  
&&
 u_2u_{3}=\alpha u_{3} u_2+\beta u_{2}+\gamma u_{3}-\beta \gamma,\\\nonumber &
&u_2 u_4=u_4 u_2-\gamma u_3+\gamma u_1,\\\nonumber  &&
u_3u_{4}=\alpha u_{4} u_3+\beta u_{4}+\gamma u_{3}-\beta \gamma;\\ 
 M\ge 5:&&   u_n u_{n+1}=\alpha u_{n+1}u_n,\label{v5}\\&& u_n u_m=u_mu_n,\ \
|n-m|>1,\ n,m\in\Z_M.\nonumber
  \end{eqnarray}
The constants $\alpha,\beta,\gamma,\eta\in\C, \ \alpha\ne 0$ are arbitrary.
\end{theorem}
\begin{proof} When $M=3$, the ideal $\fI_3$ is generated by three polynomials 
$\ff_{1,2}, \ff_{1,3}$ and $\ff_{2,3}$. We differentiate them by the derivation 
$\partial_{t_1}$ associated to 
the Volterra equation (\ref{vol3}) and project it on the quotient algebra. We 
have
\begin{eqnarray*}
&&\pi_{\fI_3} \left(\partial_{t_1}(\ff_{1,2})\right) =\omega_{1,2} 
(\omega_{1,3} \omega_{2,3}-1) u_3 u_2 u_1 + (\sigma_{1,2}^2+\omega_{1,2} 
\sigma_{1,3}^2) u_2^2
+(\omega_{1,2} \omega_{1,3} \sigma_{2,3}^1-\sigma_{1,2}^1) u_1^2\\
&&\ +(\omega_{1,2} \omega_{2,3} \sigma_{1,3}^3  +\omega_{2,3} \sigma_{1,2}^3 + 
\sigma_{1,2}^3 - \sigma_{1,2}^2) u_3 u_2
+(\omega_{1,2} \omega_{1,3} \sigma_{2,3}^3  +\omega_{1,3} \sigma_{1,2}^1 - 
\omega_{1,3} \sigma_{1,2}^3 - \sigma_{1,2}^3) u_3 u_1\\
&&\ +\omega_{1,2} (\omega_{1,3} \sigma_{2,3}^2+\sigma_{1,3}^1)  u_2 u_1 
+(\omega_{1,2} \sigma_{1,3}^3 \sigma_{2,3}^3+\sigma_{1,2}^1 
\sigma_{1,3}^3-\sigma_{1,2}^3 \sigma_{1,3}^3+\sigma_{1,2}^3 \sigma_{2,3}^3) 
u_3\\
&&\ +(\omega_{1,2} \sigma_{1,3}^3 \sigma_{2,3}^2+\omega_{1,2} 
\eta_{1,3}+\sigma_{1,2}^1 \sigma_{1,3}^2-\sigma_{1,2}^3 
\sigma_{1,3}^2+\sigma_{2,3}^2 \sigma_{1,2}^3+\eta_{1,2}) u_2\\
&&\ +(\omega_{1,2} \omega_{1,3} \eta_{2,3}+\omega_{1,2} \sigma_{1,3}^3 
\sigma_{2,3}^1+\sigma_{1,2}^1 \sigma_{1,3}^1-\sigma_{1,2}^3 
\sigma_{1,3}^1+\sigma_{2,3}^1 \sigma_{1,2}^3-\eta_{1,2}) u_1\\
&&\ +((\omega_{1,2} \sigma_{1,3}^3 \eta_{2,3} +\sigma_{1,2}^1 \eta_{1,3} 
-\sigma_{1,2}^3 \eta_{1,3}+\sigma_{1,2}^3 \eta_{2,3}) .
\end{eqnarray*}
In the same way, we compute $\pi_{\fI_3} 
\left(\partial_{t_1}(\ff_{2,3})\right)$ and $\pi_{\fI_3} 
\left(\partial_{t_1}(\ff_{1,3})\right)$. If $\fJ_3$ is preserved under the 
derivation $\partial_{t_1}$, all coefficients
in these expressions should vanish, which leads to an algebraic system for 
$\omega_{i,j},\sigma^r_{i,j},\eta_{i,j}, 1\le i<j\le 3$ and $r\in\{1, 2,3\}$. 
The only nontrivial solution of this system is
\begin{eqnarray*}
\omega_{1,2}=\omega_{2,3}=\frac{1}{\omega_{1,3}};\qquad  
\sigma_{1,2}^r=\sigma_{2,3}^r=-\omega_{1,2} \sigma_{1,3}^r, \ \ r=1, 2, 
3;\qquad  \eta_{1,2}=\eta_{2,3}=-\omega_{1,2} \eta_{1,3},
\end{eqnarray*}
which is the ideal presented in the statement by setting $\omega_{1,2}=\alpha$, 
$\sigma_{1,2}^1=\beta$ and $\eta_{1,2}=\eta$.

The proof of the statement for the case when $M=4$ is similar and we do not 
present it here.  Let us now prove the last part of the statement concerning 
the case $M \geq 5$. 
The condition $M \geq 5$ implies that $u_{n+2}, u_{n+1}, u_n, u_{n-1}, u_{n-2}$ 
are algebraically independent in $\fA_M / \fI_M$ for all $n \in \mathbb{Z}$.
In the quotient algebra $\fA_M / \fI_M$,  
$\pi_{\fI_M}\left(\partial_{t_1}(\ff_{i,j})\right) = 0$ for all $i < j$ is 
equivalent to all terms with the same degree vanishing. We denote its cubic 
terms as 
$Q_{i,j}^{(3)}$. Note that
the cubic terms of $\partial_{t_1}(\ff_{i,j})$ are
\begin{eqnarray} 
&&u_{i+1}u_iu_j-u_iu_{i-1}u_j+u_iu_{j+1}u_j-u_iu_ju_{j-1}\nonumber\\
&&\qquad -\omega_{i,j}\left(u_{j+1}u_ju_i -  u_ju_{j-1}u_i +  u_ju_{i+1}u_i -  
u_ju_iu_{i-1}\right). \label{eq1}
\end{eqnarray}
It is clear that $Q_{n,n+1}^{(3)} = 0$ if and only if $\omega_{n,n+2} = 1$ for 
all $n$. We have
\begin{eqnarray*} 
&&Q_{n,n+2}^{(3)}=(\omega_{n+1,n+2}-\omega_{n,n+1}) u_{n+2} u_{n+1}u_n
 +(\omega_{n,n+3}-1)u_{n+3}u_{n+2}u_n\\&&\qquad \qquad + (1-\omega_{n-1,n+2})  
u_{n+2}u_nu_{n-1},
\end{eqnarray*}
which vanishes when
 $\omega_{n,n+3} = \omega_{n-1,n+2} = 1$ and $\omega_{n,n+1}=\omega_{n+1,n+2}$. 
We set $\omega_{n,n+1}=\alpha$. 

Let $k$ be the distance between $i$ and $j$ modulo $M$. If $k>2$, the sets 
$\{i+1,i,j \},$ $\{i,i-1,j \}$, $\{i,j+1,j\}$ and  $\{i,j,j-1\}$ are all 
distinct (elements are taken modulo $M$).
It follows from (\ref{eq1}) that, for $k >2$,
\begin{eqnarray*} 
&&Q_{i,j}^{(3)}=\omega_{i,j} ( (\omega_{i+1,j}-1) u_j 
u_{i+1}u_i-(\omega_{i-1,j}-1) u_j u_iu_{i-1})\\
&&\qquad \quad +\omega_{i,j} ( (\omega_{i,j+1}-1) u_{j+1}u_j 
u_i-(\omega_{i,j-1}-1) u_ju_{j-1} u_i) 
\end{eqnarray*}
implying that $\omega_{i+1,j}=\omega_{i,j+1}=1$ for all $i$ and $j$. This leads 
to $\omega_{i,j}=1$ for all $i$ and $j$.
So far we have proved that $\omega_{n,n+1} = \alpha$ for all $n$ and 
$\omega_{i,j} = 1$ otherwise. 

We are now ready to look at the rest terms in 
$\pi_{\fI_M}\left(\partial_{t_1}(\ff_{i,j})\right)$. The condition  
$\pi_{\fI_M}\left(\partial_{t_1}(\ff_{n,n+1})\right)= 0$ is equivalent
to the following equation (we imply sums over $r$):
\begin{eqnarray}
&&\pi_{\fI_M} \left(\sigma_{n,n+1}^r(u_{r+1}u_r-u_ru_{r-1})\right) = 
\pi_{\fI_M} 
\left(\sigma_{n,n+1}^r(u_{n+1}+u_{n+2})u_r - \sigma_{n-1,n+1}^r u_nu_r \right) 
\nonumber \\
&&\qquad \qquad+\pi_{\fI_M} \left( \sigma_{n,n+2}^r 
u_ru_{n+1}-\sigma_{n,n+1}^ru_r(u_n+u_{n-1})\right) \nonumber\\
&&\qquad \qquad+ \eta_{n,n+1}(u_{n+2}+u_{n+1}-u_n-u_{n-1})  + 
\eta_{n,n+2}u_{n+1}-\eta_{n-1,n+1}u_n. \label{eq2m}
\end{eqnarray}
In this expression, if we look at quadratic terms not containing  $u_l$, 
$n-1\le 
l\le n+2$ as a factor, we get $\sigma_{n,n+1}^r= 0$ if $r \notin \{n-1, 
n,n+1,n+2\}$.
We substitute them into (\ref{eq2m}) and get 
$\sigma_{n,n+1}^{n-1}=\sigma_{n,n+1}^{n+2}=0$ after comparing to the quadratic 
terms in its both sides.
We denote the sum over $r$ of $\sigma_{n-1,n+1}^r u_r$ by $\Sigma_n$.  The 
quadratic terms in (\ref{eq2m}) becomes
\begin{eqnarray*}
&& 0 
=\sigma_{n,n+1}^{n+1}u_{n+1}^2-\sigma_{n,n+1}^{n+1}u_{n+1}u_{n-1}+\Sigma_{n+1} 
u_{n+1}  +\sigma_{n,n+1}^nu_{n+2}u_n
-u_n\Sigma_{n} -\sigma_{n,n+1}^n u_n^2,
\end{eqnarray*}
which implies that $\Sigma_n$ is proportional to $u_n$, and further leads to 
$\sigma_{n,n+1}^{n+1}=\sigma_{n,n+1}^{n}=\Sigma_{n}=0$.
Finally from the vanishing of linear terms in \eqref{eq2m} we have 
$\eta_{n,n+1} 
= \eta_{n,n+2} = 0$. Thus we have that for all $n$, $\ff_{n,n+1} = u_nu_{n+1}- 
\alpha u_{n+1}u_n$ and $\ff_{n,n+2} = u_nu_{n+2}-u_{n+2}u_n$.

We will prove that $\ff_{n,n+m} = u_nu_{n+m}-u_{n+m}u_n$ for $m>2$ by induction.
Assume that we have for all $2 \leq l \leq k$ that $\ff_{n,n+l} = 
u_nu_{n+l}-u_{n+l}u_n$. We now compute $\partial_{t_1}(\ff_{n,n+k})$. Using the 
induction assumption we have
\begin{eqnarray*}
&&0=\pi_{\fI_M}\left(\partial_{t_1}(\ff_{n,n+k})\right)=\pi_{\fI_M}\left( 
u_nu_{n+k+1}u_{n+k} -u_{n+k+1}u_{n+k}u_n+u_{n+k}u_nu_{n-1}-u_nu_{n-1}u_{n+k} 
\right)\\
&&\quad=\sigma_{n,n+k+1}^ru_ru_{n+k}-\sigma_{n-1,n+k}^r u_nu_r 
-\eta_{n-1,n+k}u_n+\eta_{n,n+k+1}u_{n+k}.
\end{eqnarray*}
Thus the coefficient $\sigma_{n,n+k+1}^r$ should be zero whenever $r$ is not 
$n$ but also whenever $r$ is not $n+k+1$ hence the $\sigma$'s are identically 
zeros, from which it follows
that $\eta_{n,n+k+1}=0$. Hence we conclude the induction and complete the proof.
\end{proof}

Note that the proof for the case $M\geq 5$ can be directly generalized to the 
non-periodic case which means that the ideal $\fJ$ is the only stable ideal for 
the nonabelian
Volterra flow within the class of ideals where $\ff_{i,j}$ has the form 
(\ref{fijm}). This justifies our choice of the ideal $\fI$ (\ref{ideal0}) in  
the case of infinite Volterra chain (\ref{vol}). 

\subsection{Bi-quantum structure of the periodic Volterra system with period $3$}\label{bsec}

In the classical commutative case the $M=3$ periodic Volterra system 
(\ref{vol3}) is bi-Hamiltonian \cite{suris03}. There are two compatible Poisson brackets
defined by
\[
 \{u_{n+1},u_{n}\}_0=1,\qquad \{u_{n},u_{n+1}\}_1=u_{n+1} u_n,\qquad n\in\bbbz_3
\]
such that a linear combination of the Poisson brackets, called a Poisson pencil, 
\[\{\cdot,\cdot\}_{\kappa}=(1-\kappa) \{\cdot,\cdot\}_0+\kappa \{\cdot,\cdot\}_1\]
is also a Poisson bracket for any choice of $\kappa$, i.e. the bracket $\{\cdot,\cdot\}_{\kappa}$ is skew-symmetric and satisfies the Jacobi 
identity. The system admits two first integrals 
\begin{equation}\label{HH3}
 H_1=u_1+u_2+u_3,\qquad H_2=u_3u_2u_1,
\end{equation}
such that equations (\ref{vol3}) with commutative variables can be written in a bi-Hamiltonian form
\begin{equation}\label{biV}
 \partial_{t_1}(u_k)=\{u_k,H_2\}_0=\{u_k,H_1\}_1,\qquad k\in\bbbz_3.
\end{equation}
These first integrals Poisson commute with each other and moreover, $H_1$ is in the kernel of 
the first Poisson bracket (is a Casimir element), while $H_2$ is in the kernel of the second one
\[
 \{u_k,H_1\}_0=\{u_k,H_2\}_1=0,\qquad k\in \bbbz_3.
\]
and $H_\kappa=(1-\kappa)H_1-\kappa H_2$ is a Casimir element of the bracket $\{\cdot,\cdot\}_{\kappa}$.

According Proposition \ref{omegadag} and Theorem \ref{propM},
the periodic Volterra system  (\ref{vol3}) on the free algebra $\fA_3$ 
admits a $\partial_{t_1}$ and $\dagger$ stable difference ideal $\fJ_{\theta,\hbar}=\langle f^{(\theta,\hbar)}_{n}; n\in\Z_3\rangle$, generated by the  polynomials 
\[
 f^{(\theta,\hbar)}_{n}=q^{-1} u_{n}u_{n+1}-q  u_{n+1}u_n-\i \theta,\qquad n\in\Z_3,\quad q=e^{\i\hbar},
\]
depending on the two real parameters $ 0\le\hbar<\pi, \theta\!\in\R$.
Thus, we have a pencil of quantised algebras $\fA^{(\theta,\hbar)}=\fA_3\diagup 
\fJ_{\theta,\hbar}$. Algebra  $\fA^{(\theta,\hbar)}$ has a central element
\[
 \cH(\theta,\hbar)= \sin (\hbar)H_2+ \theta(2+\cos(2\hbar)) H_1,
\]
where the self-adjoint elements 
\begin{eqnarray}
 H_1&=&u_1+u_2+u_3,\\ H_2&=&  
\sum_{\sigma\in\S_3}u_{\sigma(1)}u_{\sigma(2)}u_{\sigma(3)}
\\&=&
 3(q^2+1)u_3 u_2 u_1+\i\theta  \left(
 (2q +q^{-1})(u_1+u_3)-(q +2q^{-1}) u_2
 \right)\nonumber
\end{eqnarray}
are first integrals  for the quantum Volterra system
\begin{equation}\label{qvol}
 (u_{n})_{t_1}=q(u_{n+1}u_n-u_nu_{n-1}),\qquad n\in\Z_3.
\end{equation}
Moreover, system (\ref{qvol}) in algebra  $\fA^{(\theta,\hbar)}$ can be represented in the Heisenberg form
\[
 (u_{n})_{t_1}=\dfrac{\i}{2\sin\hbar }[H_1,u_n]=-\dfrac{\i}{2
\theta(2+\cos(2\hbar)) }[H_2,u_n].
\]

With two quotient algebras $\fA^{(\theta,0)}$ and $\fA^{(0,\hbar)}$ we associate 
the following {\em bi-quantum} structure (a quantum deformation of the 
bi-Hamiltonian structure (\ref{biV})) as follows:
\[
 \begin{array}{rll}
  \mbox{choice of parameters}\quad& \theta\ne0,\ \hbar=0,\ q=1 &\theta=0,\ 0<\hbar<\pi,\ q=e^{\i\hbar}\\&&\\
 \mbox{stable ideal in }\fA_3\quad & \fJ_{\theta,0}\quad &\fJ_{0,\hbar}\\&&\\
 \mbox{quantised algebra }\quad & \fA^{(\theta,0)}=\fA_3\diagup\fJ_{\theta,0}\quad & \fA^{(0,\hbar)}=\fA_3\diagup\fJ_{0,\hbar}\\&&\\
  \mbox{self-adjoint central element} \quad & H_1=u_1+u_2+u_3\quad 
&H_2=3(1+q^2)u_3u_2u_1\\&&\\
   \mbox{the Heisenberg form of (\ref{qvol})} \quad & (u_n)_{t_1}=-\dfrac{\i}{6
\theta}[H_2,u_n]\quad &(u_n)_{t_1}=\dfrac{\i}{2\sin\hbar }[H_1,u_n]\\&&\\
 \end{array}
\]

More work is required to study the quantum periodic Volterra systems with $M\ge 4$
(\ref{v4}), (\ref{v5}) as we did for $M=3$ above, which is not included in this paper.

\subsection{Quantisation of periodic reductions of the cubic symmetry 
}\label{m4sec}
In this section, we study the quantisation problem for periodical reductions 
of the cubic symmetry (\ref{voltf2}). In the infinite case this system admits 
two distinct quantisations (Proposition \ref{proV2}).

We claim that:
\begin{enumerate}
 \item In the case $M=3$ the quantisation ideal (\ref{fijm}) is generated by 
relations (\ref{v3}).
\item For odd $M\ge 5$  the quantisation ideal (\ref{fijm}) is generated by 
relations (\ref{v5}).
\item For even $M\ge 6$ there are two distinct quantisations corresponding to 
the ideal $\fI_a$ generated by the relations (\ref{v5}) and $\fI_b$ generated 
by relations 
\begin{equation}\label{ibm}
 u_nu_{n+1}=(-1)^n \omega u_{n+1}u_n,\qquad  
u_nu_m+u_mu_n=0\ \ 
\mbox{if}\ \ |n-m|\geqslant 2,\quad n,m \in\bbbz_M\, .
\end{equation}

\end{enumerate}
  
The case $M=4$ is exceptional, it admits three distinct quantisation ideals. 
One quantisation ideal is generated by by commutation relations (\ref{v4}) and 
the other two are generated by homogeneous quadratic commutation relations.
The periodical reduction of the system (\ref{voltf2}) with the period $M=4$ can 
be written in the form (Here we also add the constant $q^2$ following (\ref{qvolh})):
\begin{equation}\label{voltf24}
 \partial_{t_2}u_n=q^2 \left(u_{n+2} u_{n+1} u_n + u_{n+1}^2 u_n+u_{n+1} u_n^2 - u_n^2
u_{n+3}-u_n u_{n+3} u_{n+2}-u_n u_{n+3}^2\right),  q=e^{\i \hbar}
\end{equation}
where the lower index $n\in \bbbz_4$.
In the free algebra $\fA_4=\C\langle u_1,\ldots u_4\rangle$ we consider the 
ideal $\fJ$ 
\begin{equation}\label{ideal04}
 \fJ=\langle \ff_{i,j}\,;\, 1\leqslant i<j \leqslant 4\rangle,\qquad 
 \, \ff_{i,j}=u_i u_j-\omega_{i,j}u_j u_i,
\end{equation}
generated by six homogeneous quadratic polynomials $\ff_{i,j}$, which depend on 
six nonzero constants $\omega_{i,j}$.
The ideal $\fJ$ is $\partial_{t_2}$--stable if and only if 
$\partial_{t_2}(\ff_{i,j})\in\fJ,\  1\leqslant i<j \leqslant 4$. This is 
equivalent to the following system of equations on the parameters $\omega_{i,j}$
\begin{equation}\label{eqss}
 \omega _{2,4}^2=1 ,\ \ \omega_{1,4}^2 \omega _{3,4}^2=1 ,\ \ 
\omega_{2,3}=\omega _{2,4} \omega _{3,4} ,\ \ \omega_{1,2}=\omega _{1,4} \omega 
_{2,4} \omega _{3,4}^2 ,\ \ \omega_{1,3}=\omega _{1,4} \omega _{3,4}.
\end{equation}
Solving the above system of equations, we obtain the following statement:
\begin{theorem}\label{propS} A nonabelian system (\ref{voltf24}) admits a 
$\fJ$--quantisation of the form (\ref{ideal04}) if and only if the six  
constants $\omega_{i,j}$ take values as in one of four cases:
  \begin{eqnarray*}
 \begin{array}{lrrrrrr}
 &\omega _{1,2}&\omega _{1,3}&\omega _{2,3}&\omega _{1,4}&\omega _{2,4}&\omega 
_{3,4}\\
{\rm (a):}\quad &    \omega ,\ &   1,\ &   \omega ,\ &   \omega^{-1} ,\ &   1 
,\ &   \omega;\\
{\rm (b):}\quad &    \omega ,\ &   -1,\ &   -\omega ,\ &   -\omega^{-1} ,\ &   
-1,\ &   \omega;\\
{\rm (c):}\quad &    -\omega ,\ &   -1,\ &   \omega ,\ &   -\omega^{-1} ,\ &   
1,\ &   \omega;\\
{\rm (d):}\quad &    -\omega ,\ &   1,\ &   -\omega ,\ &   \omega^{-1} ,\ &   
-1,\ &   \omega ,
 \end{array}
  \end{eqnarray*}
and $\omega=q^2=e^{2\i\hbar}$, where $\hbar\in\R$.
Moreover, in each of the above four cases the system (\ref{voltf24}) is a 
super-integrable quantum system.
\end{theorem}

The first and second solutions correspond to the cases (a) and (b) in the 
Proposition \ref{proV2}. Solutions (c) and (d) are new, they  are related by 
the 
automorphism $\cS$ of $\fA_4$ and thus equivalent. The commutation relations in 
the case (a) can be extended by non-homogeneous terms (\ref{v4}), while 
commutation relations (b), (c) and (d) do not admit non-homogeneous extensions.

\begin{proof} First note that the four cases listed in the statement correspond 
to the four solutions of the system (\ref{eqss}). It is obvious that in each 
case the ideal is $\dagger$--stable if and only if  $\omega^\dagger=\omega^{-1}$.
Thus we can set $\omega=e^{2\i\hbar},\ \hbar\in\R$.

We now prove the super-integrability of the obtained system in each case. Let
$$
H =u_1+u_2+u_3+u_4,
$$
which is a first integral for the quantum system  (\ref{voltf24}) in all four 
cases. Moreover,
in all four cases the quantum system  (\ref{voltf24}) for self-adjoint variables 
$u_n$  can be written in the same 
Heisenberg form (\ref{qvolb}):
\begin{equation}\label{heisym2}
\partial_{t_2}(u_n)=\dfrac{\i}{2\sin(2\hbar)}[H^2,u_n],\qquad n\in\Z_4.
\end{equation}
In the case (a), corresponding to the quantisation of the Volterra system, the 
quantisation ideal $\fI_a$ is generated by the commutation relations between 
the variables $u_k$ as follows:
\begin{equation}
 \begin{array}{lll}
 u_1u_2=\omega u_2u_1,\quad&u_1u_3=u_3u_1,\quad &u_4u_1=\omega u_1u_4,\\
 u_2u_3=\omega u_3u_2,\quad&u_2u_4=u_4u_2,\quad &u_3u_4=\omega u_4u_3.
 \end{array}
\end{equation}
The algebra $\fA_4\diagup \fJ_a$ has two central elements
\[
\cH_1=u_3u_1,\qquad \cH_2=u_4u_2 .
\]
Since the central elements of the algebra commute with the Hamiltonian, 
they are first integrals of the system (\ref{heisym2}). The system of four 
equations (\ref{voltf24}) admits three commuting first integrals and therefore 
it is super--integrable.

In the case (b) the quantisation ideal $\fI_b$ is generated by the commutation 
relations between the variables $u_k$ as follows
\begin{equation}
 \begin{array}{lll}
 u_1u_2=\omega u_2u_1,\quad&u_1u_3=-u_3u_1,\quad &u_4u_1=-\omega u_1u_4,\\
 u_2u_3=-\omega u_3u_2,\quad&u_2u_4=-u_4u_2,\quad &u_3u_4=\omega u_4u_3.
 \end{array}
\end{equation}

The dynamical system (\ref{voltf24}) on $\fA_4\diagup \fI_b$ admits two first 
integrals
\[ H_1=u_3 u_1 ,\qquad  H_2=u_4 u_2 \,.
\]
Elements $H_1, H_2$ anti-commute with $H$, but  $H^2,\ H_1$ and $ H_2$ commute 
with each other. Thus the system (\ref{voltf24}) is super--integrable on 
$\fA_4\diagup \fI_b$.
Taking $H_1$ and $H_2$ as Hamiltonians we can find two commuting symmetries of 
the quantum equation (\ref{voltf24}) on $\fA_4\diagup \fI_b$, i.e.,
\[
 \partial_{\xi}(u_n)=[H_1,u_n]=2u_3u_1u_n,\qquad 
\partial_{\eta}(u_n)=[H_2,u_n]=2u_4u_2u_n.
\]
The algebra $\fA_4\diagup \fI_b$ has three central elements
\[
 \cH=u_4u_3u_2u_1,\qquad \cH_1=u_3^2u_1^2,\qquad \cH_2=u_4^2u_2^2.
\]

In the case (c), which is new, the quantisation ideal $\fI_c$ is generated by 
the commutation relations between the variables $u_k$ as follows:
\begin{equation}
 \begin{array}{lll}
 u_1u_2=-\omega u_2u_1,\quad&u_1u_3=-u_3u_1,\quad &u_4u_1=-\omega u_1u_4,\\
 u_2u_3=\omega u_3u_2,\quad&u_2u_4=u_4u_2,\quad &u_3u_4=\omega u_4u_3.
 \end{array}
\end{equation}
The dynamical system (\ref{voltf24}) on $\fA_4\diagup \fI_c$ admits the 
first integral $H_1=u_3 u_1$ commuting with $H^2$.
The algebra $\fA_4\diagup \fI_c$ has two central elements
\[
\cH_1=u_3^2u_1^2,\qquad \cH_2=u_4 u_2.
\]
The first integrals $H^2,H_1$ and $\cH_2$ are obviously independent and 
therefore system (\ref{voltf24}) on $\fA_4\diagup \fI_c$ is super--integrable.

The last case (d) can be obtained from the case (c) by the cyclic permutation 
of the variables $\{u_1,u_2,u_3,u_4\}\mapsto \{u_2,u_3,u_4,u_1\}$. 
\if
By now, we complete the proof. For each case, we explicitly present its first 
integrals and Central elements.
\fi
\end{proof}

In the case $M=5$ the only $\partial_{t_2}$--stable ideal is defined by 
(\ref{v5}). The system admits three commuting first integrals
\[
 H_1=\sum\limits_{k\in\bbbz_5}u_k,\quad 
H_2=\sum\limits_{k\in\bbbz_5}(u_k^2+u_k u_{k+1}+u_{k+1}u_k),\quad \cH=u_5 u_4 
u_3 u_2 u_1,
\]
where $\cH$ is a central element of the algebra. The Heisenberg equations 
corresponding to $H_1$ and $H_2$ results in the periodic Volterra system and 
its cubic symmetry respectively.

\section{Quantisation of the nonabelian Volterra Hierarchy}\label{secproof}
In this section, we extend Proposition \ref{proV} and Proposition \ref{proV2} 
in Section \ref{sec2} to the whole nonabelian Volterra hierarchy. We show that 
the quantum ideal $\fI_a$ (\ref{idi}) is invariant with respect to 
 every member of the hierarchy  (\ref{volh}) (Theorem \ref{main1}) and that the 
quantum ideal $\fI_b$ (\ref{idj}) is invariant with respect to
 every even member of the nonabelian Volterra hierarchy 
 $$\d_{t_{2\ell}}(u)=\cS(X^{(2\ell)})u-u\cS^{-1}(X^{(2\ell)}),\qquad \ell\in\N,$$
that is, odd degree symmetries of the nonabelian  Volterra equation (Theorem 
\ref{main2}).

We are going to use the explicit expressions given by (\ref{volh}) to prove 
these statements.
First we introduce some notations and definitions inspired by the monomials 
appearing in $X^{(l)}$. 
 
 Let $\a=(\a_1, \a_2, \cdots, \a_k) \in \bbbz^k$ be a $k$-component vector. For 
each $\a\in \bbbz^k$, we define the $k$-degree monomial $u_{\a} =u_{\a_1} 
u_{\a_2}\cdots u_{\a_k}$. We denote the degree of $\a$ by  $|\a|=k$. 
Conventionally, we write $(\a_1+1,\a_2+1, \cdots, \a_k+1)$ as $\a+1$. Thus we 
have
 $\cS^i u_{\a}=u_{\a+i}$ for $i\in\bbbz$. The number of variable $u_i$ in 
monomial $u_{\a}$ is denoted by $\nu(\a,i)$.
 Similarly, we denote by $\nu(\a,\geq i)$ the number of $k \geq i$ such that 
$u_k$ appears in $u_{\a}$, counted with multiplicities.
 We say that two monomials $u_{\a}$ and $u_{\beta}$ are \textit{similar} 
written as ${\a} \sim {\beta}$ if $\nu(\a,i)=\nu(\beta,i)$ for all $i\in\bbbz$.
 
 We introduce two sets of distinguished monomials, for $k \geq 1$
 \begin{eqnarray*}
  && \cA^k=\left\{\a\in \bbbz^k\big{|} 1-k\leq \a_{k} \leq 0,\ k-1\geq\a_1 \geq 
0,\ \a_{i+1}+1 \geq \a_i,\ i=1,...,k-1\right\};\\
  &&\cZ^k_{\geq}=\left\{\a\in \bbbz^k\big{|} \a_{i+1}+1 \geq \a_i\geq 
\a_{i+1},\ i=1,...,k-1\right\}.
 \end{eqnarray*}
 We say that a $k$-degree monomial $u_{\a}$ is admissible if $\alpha\in \cA^k$ 
and is nonincreasing if $\a\in \cZ^k_{\geq}$.

Using these notations, we can simply write the expression $X^{(k)}$ given by 
(\ref{xl}) as
\begin{equation}\label{xk}
 X^{(k)}=\sum_{\a\in \cA^k} u_{\a} .
\end{equation}
Given an ideal $\fI$, either $\fI_a$ or $\fI_b$, the canonical projection 
$\pi_{\fI}: \fA \rightarrow \fA/\fI $ acts on $X^{(k)}$ as follows:
\begin{eqnarray*}
 \pi_{\fI} (X^{(k)})=\sum_{\a\in \cA^k\cap \cZ^k_{\geq}} P^{\fI}_{\a}(\omega) 
u_{\a},
\end{eqnarray*}
where $P^{\fI}_{\a}(\omega)$ is the unique polynomial in $\mathbb{Z}[\omega]$ 
such that for $\a\in  \cA^k\cap \cZ^k_{\geq}$, 
 \begin{equation}\label{pa}
 P^{\fI}_{\a}(\omega)u_{\a} =\pi_{\fI}  \left( \sum_{\beta\in \cA^k, \beta \sim 
\a} u_{\beta}\right).
 \end{equation}
We often write it as $P_{\a}(\omega)$ if there is no ambiguity.

We say that two polynomials $f,g\in\fA$ are $\fI$--equivalent denoted by  $f\stackrel{\fI}{\simeq}g$ if $f-g\in\fI$. Polynomials $f$ and $g$ are $\fI$  equivalent if and only if $\pi_{\fI}(f)=\pi_{\fI}(g)$.

\subsection{Quantisation of the Volterra hierarchy}\label{proofm1}
In this section, we will prove that the ideal $\fI_a$ defined by \eqref{idi} is 
preserved by the symmetry flows  \eqref{volh}, 
for all $\ell\in \bbbn$.

To do so, we need to study the polynomials $P^{\fI_a}_{\a}(\omega)$. Here we 
focus on the quantum ideal $\fI_a$. For the sake of simplicity we write the 
polynomials as $P_{\a}(\omega)$, which are in 
$\mathbb{Z}_{+}[\omega]$.
For example, we have 
\begin{eqnarray*}
&& \pI(X^{(1)})=X^{(1)}=u; \qquad \pI(X^{(2)})=X^{(2)}=u_1 u+u^2+u u_{-1};\\
&&\pI(X^{(3)})=u_2 u_1 u+u_1^2 u +(1+ \om) u_1 u^2+u^3+(1+ \om) u^2 u_{-1} +u_1 
u u_{-1}+uu_{-1}^2+u u_{-1} u_{-2}.
\end{eqnarray*}
This defines the polynomials $P_{\a}(\om)$, e.g., $P_{(0,0,-1)}(\om)=1+\om$. In 
general, we prove the following identity:
\begin{Pro}\label{pro1} Let $\a\in \cZ^k_{\geq}$. Then, we have
\begin{equation} \label{eqp}
P_{\a}(\om)+ \om^{\nu(\a,0)}P_{\a-1}(\om) =P_{\a-1}(\om)+ \om^{\nu(\a,1)} 
P_{\a}(\om).
\end{equation}
\end{Pro}
\begin{proof}
First note that this formula holds whenever $\a\notin\cA^k$ or 
$\a-1\notin\cA^k$ since for $\a\in \cZ^k_{\geq}$, $\a\in \cA^k$ if and only if 
$\nu(\a,0)\neq 0$. If $\a\notin\cA^k$, then  $P_{\a}(\om)=0$ and $\nu(\a,0)=0$. 
Similarly,
if $\a-1\notin\cA^k$, then  $P_{\a-1}(\om)=0$ and $\nu(\a,1)=0$. Thus the 
formula holds in both cases.

We now assume that $\a\in\cA^k$ and $\a-1\in\cA^k$.
Consider the set $E_{\a}$ defined as
\begin{eqnarray*}
 E_{\a}=\left\{\beta\in \bbbz^k\big{|} \b\sim \a,\ \b_1\geq 0,\ \b_k\leq 1, \ 
\b_{i} \leq \b_{i+1}+1, i=1,..., k-1 \right\}.
\end{eqnarray*}
We split $E_{\a}$ in two different ways by defining four subsets of $E_{\a}$:
\begin{equation*}
\begin{split}
A_{\a}&=\{\b \in E_{\a}  \, |\,  \b_k \leq 0 \}, \qquad B_{\a}=\{\b \in E_{\a}  
\, |\,  \b_1 \geq 1 \},\\
C_{\a}&=\{\b \in E_{\a} \, |\,  \b_k =1 \}, \qquad D_{\a}=\{\b \in E_{\a}  \, 
|\,  \b_1=0 \}.
\end{split}
\end{equation*}
It is clear that $E_{\a}=A_{\a}\cup C_{\a}=B_{\a} \cup D_{\a}$, $A_{\a} \cap 
C_{\a}=\emptyset$  and $B_{\a} \cap D_{\a}=\emptyset$. We now have
\begin{eqnarray}\label{pie}
 \pI\left(\sum_{\beta\in E_{\a}} u_{\b} \right)=\pI\left(\sum_{\beta\in A_{\a}} 
u_{\b} \right)+\pI\left(\sum_{\beta\in C_{\a}} u_{\b} \right)
 =\pI\left(\sum_{\beta\in B_{\a}} u_{\b} \right)+\pI\left(\sum_{\beta\in 
D_{\a}} u_{\b} \right) .
\end{eqnarray}
We are going to evaluate each term in it. Note that $A_{\a}= \cA^k$ is the set 
of all elements equivalent to $\a$. Thus by definition (\ref{pa}), we have
\begin{equation}\label{pia}
\pI\left(\sum_{\beta\in A_{\a}} u_{\b} \right)=P_{\a}(\om) u_{\a}.
\end{equation}
For any $\beta\in B_{\a}$, we have $\beta-1\in \cA^k$ and
$\b-1 \sim \a-1$  and thus 
\begin{equation}\label{pib}
 \pI\left(\sum_{\beta\in B_{\a}} u_{\b} \right)=\cS \pI\left(\sum_{\beta-1\in 
\cA^k,\beta-1\sim \a-1} u_{\b-1} \right)
=\cS\left(P_{\a-1}(\om) u_{\a-1}\right)=P_{\a-1}(\om) u_{\a}.
\end{equation}
Let $\b\in D_{\a}$. There is $\b_i>0$ for some $0<i<k$ since $\b\sim \a$ and 
$\a-1\in \cA^{k}$. 
Assume that there are $0<m\leq k$ positive components at positions $i_1 \leq 
i_2\leq \cdots \leq i_m$ in $\b$. Starting from $i_1$, we
find the first zero entry on the left of $i_1$, that is, $l_1=\max_{1\leq j\leq 
i_1-1}\left\{\b_j=0\right\}$ and move the components from $l_1$ to $i_1-1$ to 
the right of $i_1$ and obtain $\b^1$ with 
$$\b^1_j=\b_{j}, 1\leq j\leq l_1-1;\ \b^1_{l_1}=\b_{i_1};\ \b^1_j=\b_{j-1}, 
l_1+1\leq j\leq i_1;\ \b^1_j=\b_{j}, i_1+1\leq j \leq k.$$
For $\b^1$, we find the first zero entry on the left of $i_2$, that is, 
$l_2=\max_{l_1+1\leq j\leq i_2-1}\left\{\b^1_j=0\right\}$ and move the 
components from $l_2$ to $i_2-1$ to the right of $i_2$ and obtain $\b^2$. We 
repeat this procedure for all positive components in $\b$.
Thus we obtain a $k$-component vector $\gamma=\b^l\in A_{\a}$. This leads to
\begin{equation}\label{pid}
\pI\left(\sum_{\beta\in D_{\a}} u_{\b} \right)=\pI\left(\sum_{\gamma\in A_{\a}} 
\om^{\nu(\b,1)} u_{\gamma} \right)
=\om^{\nu(\a,1)} \pI\left(\sum_{\gamma\in A_{\a}} u_{\gamma} \right)= 
\om^{\nu(\a,1)} P_{\a}(\om) u_{\a}.
\end{equation}
Similarly, let $\b\in C_{\a}$. There is $\b_i\leq 0$ for some $0<i<k$ since 
$\b\sim \a$ and $\a\in \cA^k$.
For all nonpositive components,
we move the first component being $1$ on its right to its left,
taking with all the components of $\b$ on its left that are larger than $1$. 
Thus we obtain a $k$-component vector
$\gamma\in B_{\a}$. This leads to
\begin{equation}\label{pic}
\pI\left(\sum_{\beta\in C_{\a}} u_{\b} \right)=\pI\left(\sum_{\gamma\in B_{\a}} 
\om^{\nu(\b,0)} u_{\gamma} \right)
=\om^{\nu(\a,0)} \pI\left(\sum_{\gamma\in B_{\a}} u_{\gamma} \right)= 
\om^{\nu(\a,0)} P_{\a-1}(\om) u_{\a}.
\end{equation}
We substitute (\ref{pia})-(\ref{pic}) into (\ref{pie}) and thus we obtain the 
required identity (\ref{eqp}).
\end{proof}
In the same way as the proof of Proposition \ref{pro1}, we are able to show that
\begin{equation}\label{eqpm}
 P_{\a+m}(\om)+ \om^{\nu(\a,-m)}P_{\a+m-1}(\om) =P_{\a+m-1}(\om)+ 
\om^{\nu(\a,1-m)} P_{\a+m}(\om) \quad \mbox{for all $m \in \mathbb{Z}$.}
\end{equation}
This leads to the following statement:
\begin{Cor}\label{cor1}
Let $\a\in \cZ^k_{\geq}$. There exists a non zero rational function 
$R_{\a}(\om) \in \mathbb{Q}(\om)$ such that
\begin{equation}\label{pm0}
P_{\a+m}( \om)=R_{\a}(\om)(1-\om^{\nu(\a, -m)}) \quad \mbox{for all $m \in 
\mathbb{Z}$.}
\end{equation}
\end{Cor}
\begin{proof} For $\a\in \cZ^k_{\geq}$, there exists $l\in\bbbz$ such that 
$\nu(\a+l,0)=\nu(\a, -l)\neq 0$.
By iterating (\ref{eqpm}) we get 
\begin{equation*}
P_{\a+m}( \om) (1-\om^{\nu(\a, -l)})=P_{\a+l}( \om)(1-\om^{\nu(\a,-m)}) \quad 
\mbox{for all $m \in \mathbb{Z}$.}
\end{equation*}
Hence choosing
\begin{equation*}
R_{\a}(\om)=P_{\a+l}( \om)(1-\om^{\nu(\a,-l)})^{-1},
\end{equation*}
we obtain the required result.
\end{proof}

\begin{theorem}\label{main1}
The quantisation ideal $\fI_a$ is stable with respect to every member of the Volterra hierarchy
 $\d_{t_\ell}(u)=\cS(X^{(\ell)})u-u\cS^{-1}(X^{(\ell)}),\ \ell\in\N$.
\end{theorem} 

\begin{proof} We fix $k$ and let $u_{\tau}=Q^{(k)}$ be the $(k+1)$-degree 
symmetry of the Volterra equation given by (\ref{volh}). Since 
$\mathcal{S}(\fI)=\fI$ we only need to show that 
\begin{equation*}
\pI\left(\partial_{\tau}(uu_m-\om^{\delta_{1,m}} u_mu) \right)=0, \quad m \in 
\bbbn.
\end{equation*}
This means that 
\begin{equation*}
\pI\left(Q^{(k)} u_m+u Q^{(k)}_m-\om^{\delta_{1,m}} Q^{(k)}_m 
u-\om^{\delta_{1,m}}u_m Q^{(k)} \right)=0 .
\end{equation*}
We rewrite it in terms of $X$. Here we simply drop its upper index of $X^{(k)}$.
\begin{eqnarray} 
&&\pI\left(uX_{m+1}u_m-\om^{\delta_{1,m}}X_{m+1}u_mu-uu_mX_{m-1}+\om^{\delta_{1,
m}}u_mX_{m-1}u \right.\nonumber\\
&&\qquad \left. 
+X_1uu_m-\om^{\delta_{1,m}}u_mX_1u-uX_{-1}u_m+\om^{\delta_{1,m}}u_muX_{-1}
\right) = 0. \label{idq}
\end{eqnarray}
It is clear that, for any $\a\in \cZ^k_{\geq}$, we have
\begin{eqnarray*}
&&u u_{\a} u_m \stackrel{\fI_a}{\simeq} \om^{ \nu(\a,1)-\nu(\a,-1)} u_{\a} u u_m, \\
&&u_m u_{\a} u\stackrel{\fI_a}{\simeq} \om^{\delta_{1,m}} \om^{ \nu(\a,m+1)-\nu(\a,m-1)} u_{\a} u 
u_m, \\
&&uu_m u_{\a} \stackrel{\fI_a}{\simeq}  \om^{ \nu(\a,m+1)+ \nu(\a,1)-\nu(\a,-1)-\nu(\a,m-1)} u_{\a} 
u u_m.
\end{eqnarray*}
Note that for all $l \in \mathbb{Z}$, we have
\begin{equation*}
\pI(X_l)=\pI(\cS^l X) =\cS^l \pI(X)=\cS^l\left(\sum_{\a\in \cZ^k_{\geq}} 
P_{\a}(\omega) u_{\a}\right)
=\sum_{\a\in  \cZ^k_{\geq}} P_{\a}(\omega) u_{\a+l}
=\sum_{\a\in  \cZ^k_{\geq}} P_{\a-l}(\omega) u_{\a}.
\end{equation*}
Here the sum is over all $\a\in \cZ^k_{\geq}$ including the ones not in $\cA^k$.
Hence, the left-handed side of (\ref{idq}) becomes
\begin{eqnarray*}
&&\sum_{\a\in  \cZ^k_{\geq}} \left(P_{\a-m-1}( \om) 
-P_{\a-m+1}( \om) \om^{ \nu(\a,m+1)-\nu(\a,m-1)}
\right) \left(\om^{ \nu(\a,1)-\nu(\a,-1)}-1\right) \pI(u_{\a} u u_m)\\
&&+\sum_{\a\in  \cZ^k_{\geq}} \left(P_{\a-1}(\om) -P_{\a+1}(\om) \om^{ 
\nu(\a,1)-\nu(\a,-1)}\right) \left(1-\om^{ \nu(\a,m+1)-\nu(\a,m-1)}\right)
\pI(u_{\a} u u_m)
\end{eqnarray*}
For any $\a\in \cZ^k_{\geq}$, we need to check that the coefficient of 
$\pI(u_{\a} u u_m)$ vanishes.
Using Corollary \ref{cor1}, it amounts to compute
\begin{eqnarray*}
 &&\left(1-\om^{\nu(\a, m+1)}-(1-\om^{\nu(\a, m-1)})\om^{ 
\nu(\a,m+1)-\nu(\a,m-1)}
\right) \left(\om^{ \nu(\a,1)-\nu(\a,-1)}-1\right)\\
&&+\left(1-\om^{\nu(\a, 1)} -(1-\om^{\nu(\a, -1)}) \om^{ 
\nu(\a,1)-\nu(\a,-1)}\right) \left(1-\om^{ \nu(\a,m+1)-\nu(\a,m-1)}\right),
\end{eqnarray*}
which equals zero after the simplification and thus we complete the proof.
\end{proof}

\subsection{Non-deformation quantisation for all odd-degree Volterra 
symmetries}\label{proofm2}
In this section, we will prove that all odd-degree symmetries of the nonabelian 
Volterra hierarchy admit the quantisation $\fI_b$, that is, the ideal $\fI_b$ 
defined by \eqref{idj} is preserved by the symmetry flows  \eqref{volh} when $\ell$ is even.
We extend the automorphism $\cS$ and the
antiautomorphism $\cT$ to the algebra $\fA[\omega]$ by letting $\cS(\omega) = 
\cT(\omega)=- \omega$ so that these operators are well-defined on the quotient 
$\fA / \fI_b$.

The ideas guiding the proof essentially are the same as in the previous section 
with the notable difference of the equivalence of Proposition \ref{pro1}, which 
is much harder in this case. 

As in the previous section, for an ideal $\fI_b$, we define uniquely 
$P_{\a}(\om) \in \mathbb{Z}[\om]$ by the canonical projection $\pJ: \fA 
\rightarrow \fA/\fI_b $ acting on $X^{(k)}$.
For example, we have 
\begin{eqnarray*}
&& \pJ(X^{(1)})=X^{(1)}=u; \qquad \pJ(X^{(2)})=X^{(2)}=u_1 u+u^2+u u_{-1};\\
&&\pJ(X^{(3)})=u_2 u_1 u+u_1^2 u +(1+\om) u_1 u^2+u^3+(1- \om) u^2 u_{-1} +u_1 
u u_{-1}+uu_{-1}^2+u u_{-1} u_{-2}.
\end{eqnarray*}
This leads to the polynomials $P_{\a}(\om)$, e.g., $P_{(0,0,-1)}(\om)=1-\om$.

To prove that the ideal $\fI_b$ defined by (\ref{idj}) is preserved by the 
symmetry flows $Q^{(2k)}$, 
we first prove the equivalents of Proposition \ref{pro1} only in this case for 
$\a\in \cZ^{2k}_{\geq}$. 
 We now assume that $\a\in\cA^{2k}$ and $\a-1\in\cA^{2k}$.
In the same way as we prove Proposition \ref{pro1}, we define the set $E_{\a}$ 
as
\begin{eqnarray*}\label{setE}
 E_{\a}=\left\{\beta\in \bbbz^{2k}\big{|} \b\sim \a,\ \b_1\geq 0,\ \b_{2k}\leq 
1, \ \b_{i} \leq \b_{i+1}+1, i=1,..., 2k-1 \right\},
\end{eqnarray*}
and split $E$ in two different ways by defining four subsets of $E_{\a}$:
\begin{equation*}\label{abcd}
\begin{split}
A_{\a}&=\{\b \in E_{\a}  \, |\,  \b_{2k} \leq 0 \}, \qquad B_{\a}=\{\b \in 
E_{\a}  \, |\,  \b_1 \geq 1 \},\\
C_{\a}&=\{\b \in E_{\a}  \, |\,  \b_{2k} =1 \}, \qquad D_{\a}=\{\b \in E_{\a}  
\, |\,  \b_1=0 \}.
\end{split}
\end{equation*}
It follows that
\begin{eqnarray}\label{pie2}
 \pJ\left(\sum_{\beta\in E_{\a}} u_{\b} \right)=\pJ\left(\sum_{\beta\in A_{\a}} 
u_{\b} \right)+\pJ\left(\sum_{\beta\in C_{\a}} u_{\b} \right)
 =\pJ\left(\sum_{\beta\in B_{\a}} u_{\b} \right)+\pJ\left(\sum_{\beta\in 
D_{\a}} u_{\b} \right) .
\end{eqnarray}
We need to evaluate each term under the ideal $\fI_b$. Since $A_{\a} = 
\cA^{2k}$ is the set of all elements equivalent to $\a$, it follows from 
(\ref{pa}) that
\begin{equation}\label{pia2}
\pJ\left(\sum_{\b\in A_{\a}} u_{\b} \right)=P_{\a}(\om) u_{\a}.
\end{equation}
For any $\beta\in B$, note that $\beta-1\in \cA^{2k}$ and
$\b-1 \sim \a-1$  and thus 
\begin{equation}\label{pib2}
 \pJ\left(\sum_{\beta\in B} u_{\b} \right)=\pJ \cS \left(\sum_{\beta-1\in 
\cA^{2k},\beta-1\sim \a-1} u_{\b-1} \right)
=P_{\a-1}(-\om) \cS u_{\a-1}=P_{\a-1}(-\om) u_{\a}.
\end{equation}
We are now left to evaluate the terms for $D_{\a}$ and for $C_{\a}$ and we do 
so in Proposition \ref{proa1} and Proposition \ref{proa2}, respectively. 
\begin{Pro}\label{proa1}
Let $u_{\alpha} = u_{\mu} u^n u_{\gamma}$, where $\a=(\mu,0, \cdots,0,\gamma) 
\in \bbbz^{2k}_{\geq}$. Then we have 
\begin{equation} 
\pJ\left(\sum_{\beta \in D_{\a}} u_{\beta} \right)= (-1)^{ \nu(\alpha,\geq 2)} 
\omega^{\nu(\alpha,1)} P_{\alpha}(\omega) u_{\alpha}. 
\label{pid2}
\end{equation}
\end{Pro}
\begin{proof} We divide $\mu$ and $\gamma$ into $n$ parts and denote each part 
by $a_i$ for $\mu$ and $b_i$ for $\gamma$, where $i=1, 2, \cdots, n$, 
such that $\vec{a}=(a_1,...,a_n)\sim \mu$  and $\vec{b} = (b_1,... ,b_n)\sim 
\gamma$. Note that it is possible that the length of some $a_j$ (and/or $b_j$) 
is zero, in which case we take the convention $u_{a_j}=1$, $|a_j|=0$.
Clearly we have
$$
p=(0,b_1, a_1,0,b_2,a_2\cdots, 0, b_n,a_n)\in D_{\a}; \qquad q=(a_1,0,b_1, 
a_2,0,b_2\cdots, a_n, 0, b_n)\in A_{\a}.
$$
Thus in the quotient algebra, we obtain
\begin{eqnarray*}
&&\pJ(\prod_{i = 1}^n { u u_{b_i} u_{a_i}})= \prod_{i=1}^n{(-1)^{\nu(a_i,\geq 
2) + |a_i| |b_i|} \omega^{\nu(a_i,1)} u_{a_i} u u_{b_i}}\\
&&\qquad \qquad=  {\omega}^{ \nu(\mu,1)} (-1)^{ \nu(\mu,\geq 2)} (-1)^{ \sum_{i 
= 1}^n {|a_i||b_i| }}\prod_{i = 1}^n {u_{a_i} u u_{b_i}} .
 \end{eqnarray*}
We denote $ \sum_{i = 1}^n {|a_i||b_i| }$ by $\vec{a}\cdot\vec{b}$ and note 
that $\nu(\mu,1)=\nu(\alpha,1)$ and $ \nu(\mu,\geq 2)= \nu(\alpha,\geq 2)$. 
Hence
\begin{eqnarray*} 
&&\qquad \pJ(\sum_{p \in D_{\a}} u_p)= \pJ( \sum_{(\vec{a},\vec{b})} { \prod_{i 
= 1}^n { u u_{b_i} u_{a_i}}}) \\
&& =  {\omega}^{ \nu(\alpha,1)} (-1)^{ \nu(\alpha,\geq 2)} 
\pJ\left(\sum_{\vec{a}\cdot \vec{b}= 0\!\!\!\!\! \mod 2} { \prod_{i = 1}^n 
{u_{a_i} u u_{b_i} }}- \sum_{ \vec{a}\cdot \vec{b} =1\!\!\!\!\!\mod 2} 
{ \prod_{i = 1}^n {u_{a_i} u u_{b_i}}} \right) \\
&& =  {\omega}^{ \nu(\alpha,1)} (-1)^{ \nu(\alpha,\geq 2)} \pJ\left(\sum_{q\in 
A_{\a}} { \prod_{i = 1}^n {u_{a_i} u u_{b_i} }}- 2 \sum_{ \vec{a}\cdot \vec{b} 
=1\!\!\!\!\!\mod 2} 
{ \prod_{i = 1}^n {u_{a_i} u u_{b_i}}} \right).
\end{eqnarray*}
Note that the first term gives us the required identity (\ref{pid2}) using 
(\ref{pia2}).
Thus we are left to prove that  
$$  \pJ\left( \sum_{ \vec{a}\cdot \vec{b} =1\!\!\!\!\!\mod 2} 
{ \prod_{i = 1}^n {u_{a_i} u u_{b_i}}} \right) = 0.$$
From now on, we identify a pair of vectors $(\vec{a},\vec{b})$ with $ \prod_{i 
= 1}^n {u_{a_i} u u_{b_i}}$.
Let $$\xx = \{(\vec{a},\vec{b}) , \, \, \vec{a}\cdot\vec{b}= 1\!\!\mod 2 \}.$$ 
We split this set in two equal parts $Y$ and $Z$ after the following remarks.  
Let $c$ be the number of indices $i$ such that $|a_i|$ and $|b_i|$ are both 
odd, $d$ the number of indices such that $|a_i|$ and $|b_i|$ are both even. 
When none of this is true, the parity of $|a_i|+ |b_i|$ is odd. 

Since the length of $\alpha$ is even , the parity of $|\mu|+ |\gamma|$ is the 
same as $n$. Hence,
\begin{eqnarray*}
n  =  \sum_{i = 1}^n {|a_i| + |b_i| } \mod 2
    = n -c-d \mod 2 ,
\end{eqnarray*}
which implies that $c+d$ is even. Moreover, we know that $\vec{a}\cdot\vec{b}$ 
is odd, that is,
 \begin{eqnarray*}
1=\sum_{i = 1}^n {|a_i||b_i| } \mod 2
   = c \mod 2.
\end{eqnarray*}
Thus we have that both $c$ and $d$ are odd. 

Let $\mathcal{I} = \{i_1,..., i_{c+d} \}$ be the set of indices $i$ such that 
$|a_i|+|b_i|$ is even (We know that this set has cardinal $c+d$). Let $l$ be 
minimal so that $|a_{i_l}|$ and $|a_{i_{c+d+1-l}}|$ 
have different parity. Such $l$ exists and is unique. Indeed, if it did not
exist we would have $|a_{i_l}| \equiv |a_{i_{l+c+d-1}}|$ for all $l$ implying 
that $c$ and $d$ are even. 

We denote $i_l$ by $k(\vec{a},\vec{b})$ and $ i_{c+d+l-1}$ by $m(\vec{a}, 
\vec{b})$. However, in the sequel we will abuse notation and simply write $k$ 
and $m$, 
knowing that we have fixed the element $(\vec{a},\vec{b})$ in the set $\xx$.
Based on these definitions, we put the pair $(\vec{a}, \vec{b})$ in the set $Y$ 
if $|a_k|$ is odd and we put it in $Z$ if $|a_k|$ is even.

Let $q\in Y$ and $u_q=\prod_{i = 1}^n {u_{a_i} u u_{b_i}}$.
We are going to construct a bijective map $\phi: Y \mapsto Z$ such that 
$\phi(u_q) \stackrel{\fI_b}{\simeq} -u_q$ in the quotient algebra for all $q \in Y$. 
Define $$\phi(u_q) = (\xi_{m-1} ... \xi_k)(u_q), $$
where the maps $\xi_i$ are defined in Lemma \ref{lem3} in Appendix. Thus $\phi$ 
only transforms the product from the block $k$ to the block $m$, i.e.,
$   \prod_{i = k}^m {u_{a_i} u u_{b_i}}.$ 

By definition of the maps $\xi_i$, if we represent  $\phi(u_q)$ as $(\vec{c}, 
\vec{d})$ we see that $c_k$ and $d_k$ will have even length and that $c_m$ and 
$d_m$ will have odd length. 
It means that $\phi(u_q)$ is an element of $Z$, but also that we still have 
$k(\phi(u_q)) = k$ and $m(\phi(u_q))=m$. That is because we have left the first 
$k-1$ blocks and the last $n-m$ blocks intact. 
Since the values of $k$ and $m$ are unchanged by $\phi$ and that all the 
$\xi_i$'s are bijections, it follows that $\phi$ is a bijection as well. So it 
only remains to check that $\phi(u_q) \stackrel{\fI_b}{\simeq} -u_q$. 
By Lemma \ref{lem3} we have 
$$\phi(u_q) \stackrel{\fI_b}{\simeq} (-1)^{\eta} u_q$$
with 
$$\eta = |b_k| + |a_{k+1}| +  |b_{k+1}| + 1 + |a_{k+2}| + ... +  
|b_{m-1}|+1+|a_m|$$
We know that $|b_k|=1\mod 2$ and $|a_m|= 0\mod 2$. Hence,
$$\eta = 1 + \sum_{i = k+1}^{m-1}\left(|a_i|+|b_i|+1\right)\mod 2=1\mod 2$$
since there is a even number of indices $i$ for which $|a_i|  \equiv |b_i|$  
between $k$ and $m$. 
\end{proof}
Below we give an example to illustrate this proposition.
\begin{Ex}
Let $\alpha =(1,1,0,0,0,\text{-}1)$. We write as $\alpha= 11000\text{-}1$ for 
short. There are $18$ elements in the set $A_{\a}$. Indeed, to get an 
admissible monomial equivalent to $\a$ 
one needs to pick an element in $$\{11000, 10100, 10010, 01100, 01010, 00110 \} 
$$ and an element in $$\{ 000\text{-}1, 00 \text{-}1 0, 0 \text{-}1 00 \}.$$ 
Under the ideal $\fI_b$, we have 
$$P_{\a}(\omega) = 1 + 2 \omega^2 +2 \omega^4 + \omega^6.$$
Similarly there are $18$ elements in $D_{\a}$ since they are determined by the 
choice of an element in $\{ 01100, 01010, 01001, 00110, 00101, 00011 \}$ and an 
element in 
$\{ 000\text{-}1, 00 \text{-}1 0, 0 \text{-}1 00 \}$. So we have
$$\pJ(\sum_{\beta \in D_{\a}} u_{\beta}) = \omega^2+ 2 \omega^4 + 2 \omega^6 + 
\omega^8=\omega^2 P_{\a}(\omega),$$
which is consistent with (\ref{pid2}) since $\nu(\a,\geq 2)=0$ and $\nu(\a, 
1)=2$.

Following the line of Proposition \ref{proa1}'s proof,  with this example we 
first give a full description of the set $\xx$, then split it as $\xx = Y \cup 
Z$. 
 An admissible monomial is given by a partition of $|a_1|+|a_2|+|a_3|=2$ and a 
partition $|b_1|+|b_2|+|b_3|=1$. 
For this monomial to be in $\xx$ we need $|a_1||b_1|+|a_2||b_2|+|a_3||b_3|$ to 
be odd. It must be that $(|b_1|,|b_2|,|b_3|)$ is one of $(1,0,0)$, $(0,1,0)$ 
and $(0,0,1)$.
 Hence there are $6$ elements in $\xx$:
$$\xx = \{ 10\text{-}1100, 1010 \text{-}10, 010\text{-}110, 01010\text{-}1, 
10\text{-}1010, 10010\text{-}1 \}, $$
where $3$ elements belong to $Y$, namely, $$Y = \{ 10\text{-}1100, 1010 
\text{-}10, 10\text{-}1010 \}. $$
For each element in $Y$, we first identify the blocks $k$ and $m$, to remove a 
$1$ and a $-1$
from the block $k$ and to add them to the block $m$. We now write $Z$ in the 
same order, that is, $Z=\phi(Y)$:
  $$Z = \{ 010\text{-}110, 01010\text{-}1, 10010\text{-}1 \}. $$
One can check that $\pJ(\sum_{ \beta \in \xx} u_{\beta})= 0$ and  $ \pJ(\sum_{ 
\beta \in Y} u_{\beta})= -\pJ(\sum_{ \beta \in Z} u_{\beta})$.  
\end{Ex}

\begin{Pro}\label{proa2}
Let $u_{\alpha} = u_{\mu} u^n u_{\gamma}$, where $\a=(\mu,0, \cdots,0,\gamma) 
\in \bbbz^{2k}_{\geq}$. Then we have 
\begin{equation} 
\pJ\left(\sum_{\beta \in C_{\a}} u_{\beta} \right)= (-1)^{\nu(\a,\geq 
0)}\om^{\nu(\a,0)} P_{\a-1}(-\omega) u_{\a}. \label{pic2}
\end{equation}
\end{Pro}
\begin{proof} Note that 
 $\beta \in C_{\a}$ if and only if $\cT \cS^{-1} (\beta) \in D_{\cT(\a-1)}$, 
where $\cT$ is the antiautomorphism. Hence we have
\begin{equation*}
\cT\cS^{-1}(C_{\a}) = D_{\cT(\a-1)}.
\end{equation*}
Moreover, by definition of the map $\cT$, it is clear that
$\cT(A_{\a-1}) = A_{\cT(\a-1)}.$
Using these facts and Proposition \ref{proa1}, we obtain
\begin{eqnarray*}
&& \sum_{\beta \in C_{\a}} {u_{\beta}} = \cS \cT(  \sum_{\beta \in C_{\a}}{\cT 
\cS^{-1}(u_{\beta}) }) 
  = \cS \cT(\sum_{\beta \in D_{\cT(\a-1)}}{u_{\beta}}) \\
 &&\qquad \stackrel{\fI_b}{\simeq} \cS \cT \left( (-1)^{\nu(\cT(\a-1), \geq 2)} 
\omega^{\nu(\cT(\a-1),1)} \sum_{\beta \in A_{\cT(\a-1)}}{u_{\beta}}  \right) \\
  &&\qquad \stackrel{\fI_b}{\simeq} (-1)^{\nu(\a,\leq -1)} \omega^{\nu(\a,0)} \cS \cT  \sum_{\beta 
\in A_{\cT(\a-1)}}{u_{\beta}}  
  \stackrel{\fI_b}{\simeq} (-1)^{\nu(\a,\geq 0)} \omega^{\nu(\a,0)} \cS \sum_{\beta \in 
A_{\a-1}}{u_{\beta}}  \\
  &&\qquad \stackrel{\fI_b}{\simeq} (-1)^{\nu(\a,\geq 0)} \omega^{\nu(\a,0)} \cS 
\left(P_{\a-1}(\omega) u_{\a-1} \right)  \stackrel{\fI_b}{\simeq} (-1)^{\nu(\a,\geq 0)} 
\omega^{\nu(\a,0)} P_{\a-1}(-\omega) u_{\a},
 \end{eqnarray*}
which leads to (\ref{pic2}) since $\a\in \bbbz^{2k}_{\geq}$.
\end{proof}
Having evaluated all terms in (\ref{pie2}),
we are now in the position to prove the similar result as Proposition 
\ref{pro1} for the ideal $\fI_b$ defined by (\ref{idj}).
\begin{Pro}\label{pro2} Let $\a\in \cZ^{2k}_{\geq}$. Then, we have
\begin{equation} \label{eqp2}
P_{\a}(\om)+ (-1)^{\nu(\a,\geq 0)}\om^{\nu(\a,0)}P_{\a-1}(-\om) 
=P_{\a-1}(-\om)+ (-1)^{\nu(\a,\geq 2)}\om^{\nu(\a,1)} P_{\a}(\om).
\end{equation}
\end{Pro}
\begin{proof}
First note that this formula holds whenever $\a\notin\cA^{2k}$ or 
$\a-1\notin\cA^{2k}$ in the same reason as in the proof for Proposition 
\ref{pro1}. When $\a\in\cA^{2k}$ and $\a-1\in\cA^{2k}$,
we substitute (\ref{pia2}), (\ref{pib2}), (\ref{pid2})  and (\ref{pic2}) into 
(\ref{pie2}) and this leads to the required identity (\ref{eqp2}).
\end{proof}
Similar to Corollary \ref{cor1} for the case of ideal $\fI_a$, we have the 
following statement for the case of ideal $\fI_b$: 
\begin{Cor}\label{corj2}
Let $\a\in \cZ^{2k}_{\geq}$. There exists a non zero rational function 
$R_{\a}(\om) \in \mathbb{Q}(\om)$ such that
\begin{equation}\label{pm02}
P_{\a+m}( (-1)^m\om)=R_{\a}(\om)(1-(-1)^{\nu(\a, \geq -m)+m 
\nu(\a,-m)}\om^{\nu(\a, -m)}) \quad \mbox{for all $m \in \mathbb{Z}$.}
\end{equation}
\end{Cor}
\begin{proof} Without the loss of generality, we assume that $\a-l\in 
\cA^{2k}$, for $0\leq l\leq q$. Let
$$
R_{\a}(\om)=\frac{P_{\a}(\om)}{1- (-1)^{\nu(\a,\geq 0)}\om^{\nu(\a,0)}} .
$$
The identity (\ref{eqp2}) implies that
$$
R_{\a-1}(-\om)=R_{\a}(\om).
$$
Thus for $0\leq l\leq q$  we have
\begin{eqnarray*}
&&P_{\a-l}( (-1)^l\om)=R_{\a-l}((-1)^l\om)\left(1- (-1)^{\nu(\a,\geq l)+l 
\nu(\a,l)}\om^{\nu(\a,l)}\right)\\
&&\qquad =R_{\a}(\om)\left(1- (-1)^{\nu(\a,\geq l)+l 
\nu(\a,l)}\om^{\nu(\a,l)}\right).
\end{eqnarray*}
When $\a+m\notin \cA^{2k}$, we have $P_{\a+m}(\om)=0$ following the definition 
of (\ref{pa}).
\end{proof}

\begin{theorem}\label{main2}
The quantisation ideal $\fI_b$ is stable with respect to every even member of the Volterra hierarchy
 $\d_{t_{2\ell}}(u)=\cS(X^{(2\ell)})u-u\cS^{-1}(X^{(2\ell)}),\ \ell\in\N$.
\end{theorem}

\begin{proof}
Let $u_{\tau}=G=X^{(2\ell)}_1u-uX^{(2\ell)}_{-1}$, where $X^{(2\ell)}$ is the sum of all
admissible monomials of size $2\ell$, $\ell \geq 1$. Let $k \geq 2$. We want to show
that $\partial_{\tau}(uu_k+u_ku)$ is in the ideal $\fI_b$. By definition of 
$u_{\tau}$, this means that 
\begin{equation}\label{eqg}
\pJ(Gu_k+uG_k+G_ku+u_kG)=0,
\end{equation}
or, in terms of $X$ (we drop its upper index):
\begin{equation} \label{eq2}
\begin{split}
uX_{k+1}u_k+X_{k+1}u_ku-uu_kX_{k-1}-u_kX_{k-1}u& \\
+X_1uu_k+u_kX_1u-uX_{-1}u_k-u_kuX_{-1}&\stackrel{\fI_b}{\simeq} 0.
\end{split}
\end{equation}
Let us fix an element $\beta \in \cZ^{2\ell}$.
 We are going to show that the terms equivalent to $u_{\beta} uu_k$ modulo 
multiplication by an element of $\mathbb{Z}[\om]$ in (\ref{eq2}) cancel out. 
It is clear that 
\begin{equation*}
\begin{split}
u u_{\beta} u_k &\stackrel{\fI_b}{\simeq} (-1)^{ \nu(\beta,0)+ \nu(\beta,1)}\om^{ 
\nu(\beta,1)-\nu(\beta,-1)} u_{\beta} u u_k ,\\
u_k u_{\beta} u& \stackrel{\fI_b}{\simeq} (-1)^{ \nu(\beta,k)+ \nu(\beta,k +(-1)^k)}\om^{ 
\nu(\beta,k+1)-\nu(\beta,k-1)} u_{\beta} u_k u, \\
uu_k u_{\beta} & \stackrel{\fI_b}{\simeq}  (-1)^{ \nu(\beta,k)+ \nu(\beta,k +(-1)^k)+ \nu(\beta,0) 
+\nu(\beta,1)}\om^{ \nu(\beta,k+1)+ 
\nu(\beta,1)-\nu(\beta,0)-\nu(\beta,k-1)}u_{\beta} u u_k.
\end{split}
\end{equation*}
We know that for all $m \in \mathbb{Z}$,
\begin{equation*}
\pJ(X_{m})= \sum_{\a \in \cZ_{\geq}^{2n}} P_{\a}((-1)^{m} \om) u_{\a+m}.
\end{equation*}
Hence the $\mathbb{Z}[\om]$ coefficient of $u_{\beta} uu_k$ in 
$uX_{k+1}u_k+X_{k+1}u_ku$ is
\begin{equation*}
P_{\beta-k-1}((-1)^{k+1} \om)((-1)^{ \nu(\beta,0)+ \nu(\beta,1)}\om^{ 
\nu(\beta,1)-\nu(\beta,-1)}-1).
\end{equation*}
We compute the terms coming from $X_{-k-1}$, $X_1$ and $X_{-1}$ in a similar 
way. Thus, to prove that the coefficient of $u_{\beta} uu_k$ in (\ref{eq2}) is 
zero amounts to check that 
\begin{equation*}
\begin{split}
0 =&P_{\beta-k-1}((-1)^{k+1} \om)((-1)^{ \nu(\beta,0)+ \nu(\beta,1)}\om^{ 
\nu(\beta,1)-\nu(\beta,-1)}-1)\\
+&P_{\beta-k+1}((-1)^{k-1} \om)(-1)^{ \nu(\beta,k)+ \nu(\beta,k +(-1)^k)}\om^{ 
\nu(\beta,k+1)-\nu(\beta,k-1)}(1- (-1)^{  \nu(\beta,0) +\nu(\beta,1)}\om^{ 
\nu(\beta,1)-\nu(\beta,0)} ) \\
+& P_{\beta-1}(- \om)(1-(-1)^{ \nu(\beta,k)+ \nu(\beta,k +(-1)^k)}\om^{ 
\nu(\beta,k+1)-\nu(\beta,k-1)}) \\
+& P_{\beta+1}(-\om)(-1)^{  \nu(\beta,0)+ \nu(\beta,1)}\om^{ 
\nu(\beta,1)-\nu(\beta,-1)}((-1)^{ \nu(\beta,k)+ \nu(\beta,k +(-1)^k)}\om^{ 
\nu(\beta,k+1)-\nu(\beta,k-1)}- 1).
\end{split}
\end{equation*}
Using Corollary \ref{corj2}, we need to verify 
\begin{equation*}
\begin{split}
&(1-(-1)^{\nu(\beta,\geq k+1)+(k+1) \nu(\beta,k+1)}\om^{\nu(\beta,k+1)})((-1)^{ 
\nu(\beta,0)+  \nu(\beta,1)}\om^{ \nu(\beta,1)-\nu(\beta,-1)}-1)\\
+&(1-(-1)^{\nu(\beta,\geq k-1)+(k+1) \nu(\beta,k-1)}\om^{\nu(\beta,k-1)})(-1)^{ 
\nu(\beta,k)+ \nu(\beta,k +(-1)^k)}\om^{ \nu(\beta,k+1)-\nu(\beta,k-1)} \times\\
&(1- (-1)^{ \nu(\beta,0) + \nu(\beta,1)}\om^{ \nu(\beta,1)-\nu(\beta,0)} ) \\
+&(1-(-1)^{\nu(\beta, \geq 1)+ \nu(\beta,1)}\om^{\nu(\beta,1)}) (1-(-1)^{ 
\nu(\beta,k)+ \nu(\beta,k +(-1)^k)}\om^{ \nu(\beta,k+1)-\nu(\beta,k-1)}) \\
+&(1-(-1)^{\nu(\beta,\geq -1)+ \nu(\beta,-1)}\om^{\nu(\beta,-1)}) (-1)^{ 
\nu(\beta,0)+  \nu(\beta,1)}\om^{ \nu(\beta,1)-\nu(\beta,-1)}\times \\
&((-1)^{ \nu(\beta,k)+ \nu(\beta,k+(-1)^k)}\om^{ 
\nu(\beta,k+1)-\nu(\beta,k-1)}- 1)\\
&=0
\end{split}
\end{equation*}
and thus the identity (\ref{eqg}) holds. The proof that 
$\pJ\left(\partial_{\tau}(u_ku_{k+1}-(-1)^k\omega u_{k+1}u_k)\right)=0$ for all $k 
\in \bbbz$ is similar 
and we will not repeat it.
\end{proof}

\section{Summary and discussion}

In this paper we develop the method of quantisation of dynamical systems 
defined on free associative algebras based on the concept of quantisation 
ideals \cite{AvM20}. It enables us to determine possible commutation relations 
between the dynamical variables which are consistent with the dynamical system 
and define associative multiplication in the quotient algebra. The method does 
not use any information on the Poisson structure of the dynamical system and 
enables us to find non-deformation quantisations of the system.  To determine 
commutation relations consistent with a system is a very first 
step to its quantum theory. Next steps will require the development of the 
representation theory for the quantised algebras obtained and study  the 
spectral theory of the operators involved.

In this paper we explicitly proved that the nonabelian Volterra system
(\ref{vol}) and  its infinite hierarchy of symmetries admit the 
deformation quantisation with commutation relations (\ref{comm1}). We also 
proved that the sub-hierarchy, consisting of all odd degree symmetries, admits 
a non-deformation quantisation with commutation relations (\ref{comm2}). The 
existence of non-deformation quantisations is quite surprising. Further study 
is 
required to explore the properties of these new remarkable quantum 
algebra and quantum integrable equations.

Recently, when the paper has already been submitted to the journal, we found explicit expressions for the infinite sequence of quantum Hamiltonians $H_n$ corresponding to the $\fI_a$ quantisation of the Volterra hierarchy
\[
 H_\ell=\sum_{k\in\bbbz}\ \sum_{\alpha\in\cA_0^\ell}\frac{\omega^\ell-1}{\omega^{\nu(\alpha,0)}-1}P_\alpha(\omega)u_{\alpha+k},
\]
where $\cA_0^\ell=\{\alpha\in\cA^\ell\cap\cZ_{\geqslant}^\ell\,;\, \alpha_\ell=0\}$. Assuming that $\omega=e^{2\i\hbar},\ \hbar\in\R$, the Hamiltonians $H_\ell$ are self-adjoint $H_\ell^\dagger=H_\ell$. They commute with each other, and the dynamical equations of the quantum hierarchy can be written in the Heisenberg form (compare with (\ref{qvola})):
\[
 \partial_{t_\ell}(u_n)=\frac{\i}{2\sin(\ell\hbar)}[H_\ell,u_n],\qquad n\in\bbbz,\ \ell\in\bbbn\, .
\]
We have also  found explicit expressions for self-adjoint commuting quantum Hamiltonians  corresponding to non-deformation quantisation (\ref{comm2}) and present the quantum hierarchy with even times in the Heisenberg form. A detail proof of these results will be published elsewhere soon.

The Volterra hierarchy admits periodic reductions with any positive 
integer period $M$. We have 
shown that the Volterra system with periods $M=3,4$ admit quantisations with 
non-homogeneous commutation relations (Theorem \ref{propM}). When $M=3$, we proved the resulting quantum system is not only super integrable but also admits
bi-quantum structure, similar to its bi-Hamiltonian structure in the classical case. The cubic symmetry
of the Volterra system with period $M=4$ admits three distinct quantisations. In each case,
the quantum system is a super-integrable systems (Theorem \ref{propS}). Systems with periods $M\ge 5$ require more work, they
have not been studied in this paper in any detail.  
  
The methods developed in \cite{AvM20} and this paper can be applied to the 
nonabelian Narita-Itoh-Bogoyavlensky lattice \cite{Bog91} 
\begin{equation}\label{nib}
	u_t=\sum_{k=1}^p \left( u_{k} u- u u_{-k}\right), \quad p\in\mathbb{N}.
\end{equation}
The Volterra equation is corresponding to the case when $p=1$. Our study shows 
that system (\ref{nib}) and all equations of its hierarchy admit the 
quantisation with commutation relations 
\begin{eqnarray*}
 u_nu_{n+k}=\omega u_{n+k}u_n\, , \ \,1\leq k\leq p,\quad   
 u_nu_m=u_mu_n\, ,\quad   |n-m| >p \quad  n,m \in \mathbb{Z},
\end{eqnarray*}
where $\omega$ is a nonzero constant. The proof of this statement will be
published elsewhere. These commutation relations were also obtained by Inoue and Hikami \cite{InKa} using ultra-local Lax representation and  $R$--matrix technique.

Besides quadratic ideals, our computations for the nonabelian Volterra 
equation and its lower degree symmetries suggest that there is a
$\partial_{t_\ell}$--stable ideal generated by quadratic and cubic homogeneous 
polynomials.  For example, as far as we have checked, the first few symmetries 
in the 
 nonabelian Volterra hierarchy leave the following cubic ideal invariant:
\begin{eqnarray*}
 \tilde{\fJ}=\langle   u_nu_{n+1}u_{n-1}- u_{n+1}u_{n-1} u_{n}\, , \ 
\,u_nu_m-u_mu_n\, ;\, \  |n-m| >1 ,\  n,m \in 
\mathbb{Z}  \rangle .
 \end{eqnarray*}
Further research is needed to  study the properties of the Volterra 
chain which is well defined on the quotient algebra $\fA\diagup\tilde{\fJ}$.
Very little is known about this new invariant ideal and the quotient algebra 
which does not satisfy the condition (ii). 

The concept of quantisation ideals 
has not been linked yet with Lax representations, recursion operators, 
master-symmetries and other objects associated with the theory of integrable 
systems.  We think that further development of this theory will enable 
us to embrace a wide range of integrable systems as well as to clarify and 
simplify rather technical proofs of the statements presented in this paper.

\section*{Acknowledgments}
AVM and JPW are grateful for the support by the EPSRC small grant scheme EP/V050451/1, and partially by grants EP/P012655/1 and EP/P012698/1.
SC thanks for the National Research Foundation of Korea(NRF) grant funded by the Korea government(MSIT) (No.2020R1A5A1016126).

\section*{Appendix: Lemmas used for the proof of Proposition \ref{proa1}}
In this appendix, we are going to prove the lemmas used in constructing the 
bijection map between sets $A_{\a}$ and $D_{\a}$ (Proposition \ref{proa1}) in 
Section \ref{proofm2}.

Let $l$ be any integer. We denote by $\Lambda_l$ the set of admissible monomials 
of the form $u_a u_l u_b$  satisfying
 
 \begin{enumerate}
 \item[${\rm (i)}$] both $a$ and $b$ have components greater than $l$ if they 
are not empty.
 \item[${\rm (ii)}$] there exists a suffix $d$ of $a$ of odd length $a = cd$ 
where $c$ is either empty or ends with $l+1$. 
 \item[${\rm (iii)}$] if $b$ is non-empty then it ends with $l+1$. 
 \end{enumerate}
 If the length of $d$ in (ii) is minimal, we say that $d$ is the minimal odd 
suffix of $a$.

 We denote by $\Gamma_l$ the set of admissible monomials of the form $u_a u_l 
u_b$  where
 \begin{enumerate}
 \item[${\rm (i)}$] both $a$ and $b$ have components greater than $l$.
 \item[${\rm (ii)}$]  there exists a prefix $c$ of $b$ of odd length $b = cd$ 
where $c$ ends with $l+1$. 
 \item[${\rm (iii)}$] $b$ ends with $l+1$. 
 \end{enumerate}
If the length of $c$ in (ii) is minimal, we say that $c$ is the minimal odd 
prefix of $b$.

\begin{Lem} For all $l \in \bbbz$, we construct a bijection $\psi: \Lambda_l 
\rightarrow \Gamma_l$ such that for all $x \in \Lambda_l$,  $\pJ(\psi(x))=  
(-1)^l\omega \, \, x$.  Moreover, if $x = u_au_lu_b$ and $\psi(x) = u_cu_lu_d$, 
then $|c| =  |a| -|m|$ and $|d| = |b| + |m|$, where $m$ is the minimal odd 
suffix of $a$.
\end{Lem}
\begin{proof} We construct $\psi$ by induction on $|a|+|b|$.  The only element 
of length $2$ in $\Lambda_l$ is $u_{l+1}u_l$, while the only element of length 
$2$ 
in $\Gamma_l$ is $u_lu_{l+1}$. We let $\psi(u_{l+1}u_l) =u_l u_{l+1}$. The 
minimal odd suffix of $u_{l+1}$ is itself and we have $\pJ(u_lu_{l+1}) = 
(-1)^{l} \omega u_{l+1}u_l$, 
hence the statement of the Lemma holds for elements of length $2$. 

Suppose that we have constructed $\psi$ for all lengths strictly less than $n$ 
satisfying the statement. We now construct $\psi$ for elements of length $n$ and 
prove it satisfies the statement.
Let $u_au_lu_b$ be an element of $\Lambda_l$ of length $n$. Let $d$ be the 
minimal odd suffix of $a$. Explicitly, this $u_d$ has the form $u_eu_{l+1} 
u_{d_1}u_{l+1}... u_{d_p} u_{l+1}$,
where the $|d_i|$'s are odd and $|e|$ is even (hence possibly $e$ is empty). 
Note that in this decomposition of $u_d$, the elements $d_i$ and $e$ do not 
contain any $j < l+2$ and all end with $l+2$ 
(except if $e$ is empty). Hence for all $i = 1,..., p$, $u_{l+1}u_{d_i}$ is an 
element of $\Gamma_{l+1}$ whose length is strictly less than $n$. By the 
induction hypothesis, there exist $f_i$ of odd length and $g_i$ of even length 
such that 
$$ \psi^{-1}(u_{l+1} d_i)= u_{f_i} u_{l+1} u_{g_i} .$$
Note that $f_i$ does not have a proper odd suffix due to the last assertion in 
the Lemma. Recall that all elements in $f_i$ and $g_i$ are greater than $l+1$. 
The element
$ u_e\psi^{-1}(u_{l+1}u_{d_1} ) ... \psi^{-1}(u_{l+1}u_{d_p}) u_{l+1}$ is 
well-defined. It has exactly the same (odd) length as $d$ without any proper odd 
prefix and  
$$\pJ(u_e\psi^{-1}(u_{l+1}u_{d_1} ) ... \psi^{-1}(u_{l+1}u_{d_p}) u_{l+1})= 
((-1)^{l+1}\omega)^{-p} u_eu_{l+1} u_{d_1} u_{l+1}... u_{d_p} u_{l+1} .$$
We let
$$ \psi( u_a u_l u_b) = u_c u_lu_ e \psi^{-1}(u_{l+1}u_{d_1} ) ... 
\psi^{-1}(u_{l+1}u_{d_p}) u_{l+1} u_b .$$
Note that the last statement in the Lemma is satisfied. 
Let $$\chi=  u_e\psi^{-1}(u_{l+1}u_{d_1} ) ... \psi^{-1}(u_{l+1}u_{d_p}) 
u_{l+1}.$$
It has odd length and the number of $u_{l+1}$ in $\chi$ is $p+1$. Thus we have 
in the quotient algebra 
$$\pJ( u_l \chi)= (-1)^{1+(l+1)(p+1)} \omega^{p+1} \chi u_l .$$
hence
$$\pJ( u_l u_e \psi^{-1}(u_{l+1}u_{d_1} ) ... \psi^{-1}(u_{l+1}u_{d_p}) u_{l+1})=\pJ( u_l \chi)= (-1)^l \omega  u_eu_{l+1} u_{d_1}u_{l+1}... u_{d_p}u_{l+1} u_l$$
and a fortiori,
$$\pJ(\psi(u_au_lu_b)) = (-1)^l \omega u_au_lu_b.$$
We know that there are as many elements of length $n$ in $\Gamma_l$ as in $\Lambda_l$, hence it remains to check the injectivity of $\psi$ for length $n$.  Suppose that we have $\psi(u_au_lu_b) = \psi(u_{\tilde{a}} u_{l}u_{\tilde{b}})$.
In other words, we have 
\begin{eqnarray*}
&u_c u_lu_ e \psi^{-1}(u_{l+1}u_{d_1} ) ... \psi^{-1}(u_{l+1}u_{d_p}) u_{l+1} u_b = \\
& u_{\tilde{c}} u_lu_{\tilde{e}} \psi^{-1}(u_{l+1}u_{\tilde{d}_1} ) ... \psi^{-1}(u_{l+1}u_{\tilde{d}_q}) u_{l+1} u_{\tilde{b}}
\end{eqnarray*}
This equality implies that $c = \tilde{c}$ so we can simplify it slightly:
\begin{eqnarray*}
&u_ e \psi^{-1}(u_{l+1}u_{d_1} ) ... \psi^{-1}(u_{l+1}u_{d_p}) u_{l+1} u_b = 
u_{\tilde{e}} \psi^{-1}(u_{l+1}u_{\tilde{d}_1} ) ... \psi^{-1}(u_{l+1}u_{\tilde{d}_q}) u_{l+1} u_{\tilde{b}}
\end{eqnarray*}
Recall that $u_ e \psi^{-1}(u_{l+1}u_{d_1} ) ... \psi^{-1}(u_{l+1}u_{d_p}) u_{l+1}$ is the minimal odd prefix of the left hand side  and 
that $u_{\tilde{e}} \psi^{-1}(u_{l+1}u_{\tilde{d}_1} ) ... \psi^{-1}(u_{l+1}u_{\tilde{d}_q}) u_{l+1}$ is the minimal odd prefix of the right hand side. 
By unicity of the minimal odd prefix, they are equal. In particular, we have $b = \tilde{b}$ and $p = q$. Recall the definition of $f_i$ and $g_i$ such that $\psi^{-1}(u_{l+1}u_{d_i}) = u_{f_i} u_{l+1} u_{g_i}$. 
Similarly we write 
$$ \psi^{-1}(u_{l+1} u_{\tilde{d}_i}) = u_{ \tilde{f}_i} u_{l+1} u_{\tilde{g}_i}.$$
We have 
$$ u_{g_0} u_{f_1} u_{l+1}u_{g_1} u-{f_2} u_{l+1} ...   u_{f_p} u_{l+1} u_{g_p} = u_{\tilde{g}_0} u_{\tilde{f}_1} u_{l+1} u_{\tilde{g}_1} u_{\tilde{f}_2} u_{l+1} ...   u_{\tilde{f}_p} u_{l+1} u_{\tilde{g}_p},$$
where we have let $g_0 = e$ and $\tilde{g}_0 = \tilde{e}$. Therefore we have for all $i = 0,..., p-1$
$$ g_i f_{i+1} = \tilde{g}_i \tilde{f}_{i+1}.$$
Recall that both $f_{i+1}$ and $\tilde{f}_{i+1}$ are their own minimal odd suffix. Hence $f_{i+1}$ is the minimal odd suffix of $g_i f_{i+1}$ and $\tilde{f}_{i+1}$ is the minimal odd suffix of $ \tilde{g}_i \tilde{f}_{i+1}$. By unicity of the minimal odd suffix we have $f_{i+1} = \tilde{f}_{i+1}$, from where it follows that $g_{i} = \tilde{g}_i$. Hence
$$u_{l+1} u_{d_i} =\psi( u_{f_i} u_{l+1} u_{g_i}) = \psi( u_{\tilde{f}_i} u_{l+1} u_{\tilde{g}_i} )  = u_{l+1} u_{\tilde{d}_i}$$
and thus we complete the proof.
\end{proof}

Let $l$ be any integer. We denote by $\Theta_l$ the set of admissible monomials of the form $u_a u_l u_b$  where
 \begin{enumerate}
 \item[${\rm (i)}$] both $a$ and $b$ have components strictly smaller than $l$.
 \item[${\rm (ii)}$] there exists a suffix $d$ of $a$ of odd length $a = cd$ where $d$ starts with $l-1$. 
 \item[${\rm (iii)}$] $a$ starts with $l-1$. 
 \end{enumerate}
If the length of $d$ in (ii) is minimal, we say that $d$ is the minimal odd suffix of $a$.

 We denote by $\Phi_l$ the set of admissible monomials of the form $u_a u_l u_b$  where
 \begin{enumerate}
 \item[${\rm (i)}$] both $a$ and $b$ have components strictly smaller than $l$.
 \item[${\rm (ii)}$]  there exists a prefix $c$ of $b$ of odd length $b = cd$ where $d$ is either empty or starts with $l-1$. 
 \item[${\rm (iii)}$] $a$ is either empty or starts with $l-1$. 
 \end{enumerate}
 If the length of $c$ in (ii) is minimal, we say that $c$ is the minimal odd prefix of $b$.

\begin{Lem}  For all $l \in \bbbz$, we construct a bijection $\rho: \Theta_l \rightarrow \Phi_l$ such that   
$\pJ(\rho(x))=(-1)^{l+1}\omega^{-1} x$ for all $x \in \Theta_l$.  Moreover, if $x = u_au_lu_b$ and $\psi(x) = u_cu_lu_d$, 
then $|c| =  |a| -|m|$ and $|d| = |b| + |m|$, where $m$ is the minimal odd suffix of $a$. 
\end{Lem}
\begin{proof}
Take $ \rho =\cT \psi^{-1} \cT$, where $\cT$ maps $\Theta_l$ to $\Gamma_l$ and maps $\Lambda_l$ to $\Phi_l$. Let $u_a u_l u_b \in \Theta_l$. We have
 $$ \psi^{-1}(\cT(b)u_{-l} \cT(a)) \equiv (-1)^{l} \omega^{-1} \cT(b) u_{-l} \cT(a)$$
 and since $\cT(\omega) = -\omega$,
 $$\cT( \psi^{-1}(\cT(b)u_{-l} \cT(a))) \equiv (-1)^{l+1} \omega^{-1} a u_l b.$$
 Let $m$ be the minimal odd prefix of $\cT(a)$. We know that $\psi^{-1}(\cT(b) u_{-l} \cT(a)) = c u_{-l} d$ with $|c| = |\cT(b)| + |m|$ and $|d| = |\cT(a)|- |m|$. We have $\rho(a u_l b )  = \cT(d) u_l \cT(c)$. We conclude by noting that $\cT(m)$ is the minimal odd suffix of $a$.
\end{proof}

Recall that we identify an element of $\xx$, that is a pair  $(\vec{a}, \vec{b})$ such that $ \vec{a}\cdot \vec{b}= 1\!\!\!\mod 2$ with the product $ \prod_{i = 1}^n {u_{a_i} u u_{b_i}}$. We denote a subset of $X$ 
consisting of a part of $\xx$ such that $u_{a_j} u \in \Lambda_0$ and $u_{b_j} u \in \Theta_0$ for some $1\leq j \leq n$ by $\xx_j$ . We are going to construct bijections $\xi_j : \xx_j \rightarrow \xx_{j+1}$.
\begin{Lem}\label{lem3}
There exists a bijection $\xi_j : \xx_j \rightarrow \xx_{j+1}$, $1\leq j \leq n-1$,  so that 
$$ \xi_j(u_p) \stackrel{\fI_b}{\simeq} (-1)^{|a_{j+1}| + |b_j|} u_p, \qquad p\in \xx_j. $$
\end{Lem}
\begin{proof}
Let $(\vec{a}, \vec{b})$ be an element of $\xx_j$. Consider the product of block $j$ with block $j+1$, i.e.,
$$u_{a_j}uu_{b_j} u_{a_{j+1}} u u_{b_{j+1}} .$$ 
We have $a_j 0 a_{j+1} \in \Lambda_0$ and $b_j 0 b_{j+1} \in \Theta_0$. Hence there exist $\tilde{a}_j$, $\tilde{\tilde{a}}_j$, $\tilde{b}_j$, $\tilde{\tilde{b}}_j$ such that, 
\begin{eqnarray*}
 \psi(u_{a_j} u u_{a_{j+1}}) = u_{\tilde{a}_j}u u_{\tilde{\tilde{a}}_j}, \qquad
  \rho(u_{b_j} u u_{b_{j+1}})= u_{\tilde{b}_j} u u_{\tilde{\tilde{b}}_j} .
\end{eqnarray*}
From the definitions of $\rho$ and $\psi$ it follows  that $\tilde{\tilde{a}}_j 0 \in \Lambda_0$, $\tilde{\tilde{b}}_j 0 \in \Theta_0$ and
$$(|\tilde{a}_j|, | \tilde{\tilde{a}}_j|, |\tilde{b}_j|, | \tilde{\tilde{b}}_j|) = (|a_j| + 1, | a_{j+1}| +1, |b_j| + 1, | b_{j+1}| +1)\!\!\! \mod 2.$$
We now define $\xi_j: ((\vec{a}, \vec{b})  \mapsto (\vec{c}, \vec{d})$ as follows:
\begin{eqnarray*}
 & c_i = a_i \text{ and } d_i  = b_i \text{ if } i \neq j \text{ and } i \neq j+1 \\
 & c_j = \tilde{a}_j, \, \, d_j =  \tilde{b}_j, \, \, c_{j+1} = \tilde{\tilde{a}}_j, \, \text{ and } d_{j+1} = \tilde{\tilde{b}}_j.
 \end{eqnarray*}
It is clear that $(\vec{c}, \vec{d})$ is in the subset $\xx_{j+1}$. The map $\xi_j$ is a bijection since both $\psi$ and $\rho$ are bijections.
Moreover, we have
\begin{eqnarray*}
 \pJ( u_{\tilde{a}_j} u u_{\tilde{\tilde{a}}_j })= \omega u_{a_j} u u_{a_{j+1}}, \qquad
\pJ(u_{\tilde{b}_j} u u_{\tilde{\tilde{b}}_j})= - \omega^{-1} u_{b_j} u u_{b_{j+1}}.
\end{eqnarray*}
We know $\pJ(u_{a_j}u u_{b_j} u_{a_{j+1}} u b_{j+1})= (-1)^{|b_j| |a_{j+1}|} u_{a_j}u u_{a_{j+1}} u_{b_j}u u_{b_{j+1}}$.
Therefore, we obtain
\begin{eqnarray*}
\pJ( \xi_j(u_p))_{p\in X_j}= (-1)^{1+|b_| |a_{j+1}|} u_{a_1} u u_{b_1} \cdots u_{\tilde{a}_j} u u_{\tilde{\tilde{a}}_j} u_{\tilde{b}_j} u u_{\tilde{\tilde{b}}_j} \cdots u_{a_n} u u_{b_n}
 = (-1)^{|b_j| +|a_{j+1}|}u_q, 
\end{eqnarray*}
where $q=(\vec{c}, \vec{d})\in \xx_{j+1}$ and thus we complete the proof.
\end{proof}

\subsection*{Data availability statement}
Data sharing is not applicable to this article as no datasets were generated or 
analyzed during the current study.
\subsection*{Code availability statement}
Not applicable.
\subsection*{Conflict of interests}
We declare that there is no conflict of interests.

\end{document}